\documentclass[journal,comsoc]{IEEEtran}
\usepackage[T1]{fontenc}

\ifCLASSINFOpdf
\else
   \usepackage[dvips]{graphicx}
   \graphicspath{{../eps/}}
   \DeclareGraphicsExtensions{.eps}
\fi
\usepackage{amsmath}
\usepackage[cmintegrals]{newtxmath}
\usepackage{cite}
\usepackage{amsfonts}
\usepackage{booktabs}
\usepackage{multirow}
\usepackage{subfigure}
\usepackage{tabularx}

\usepackage{algorithm}
\usepackage{algorithmic}
\usepackage{amsmath}

\usepackage{amssymb}
\usepackage{color}
\usepackage{array}
\usepackage{hyphenat}
\usepackage{url}
\usepackage{ntheorem}

\newtheorem{definition}{\textbf{Definition}}
\newtheorem{proposition}{\textbf{Proposition}}
\newtheorem*{proof}{\textbf{Proof}}
\begin{document}
%

\title{Virtual Network Function Placement in Satellite Edge Computing with a Potential Game Approach}
%
%
%

\author{Xiangqiang~Gao,
        Rongke~Liu,~\IEEEmembership{Senior~Member,~IEEE,}
        and~Aryan~Kaushik,~\IEEEmembership{Member,~IEEE}
\thanks{X.~Gao and R.~Liu are with the School of Electronic and Information Engineering, Beihang University, Beijing 100191, China e-mail: (\{xggao, rongke\_liu\}@buaa.edu.cn).}
\thanks{A.~Kaushik is with the Department of Electronic and Electrical Engineering, University College London (UCL), London WC1E 7JE, United Kingdom e-mail: (a.kaushik@ucl.ac.uk).}}

\maketitle

\begin{abstract}
  Satellite networks, as a supplement to terrestrial networks, can provide effective computing services for Internet of Things (IoT) users in remote areas. Due to the resource limitation of satellites, such as in computing, storage, and energy, a computation task from a IoT user can be divided into several parts and cooperatively accomplished by multiple satellites to improve the overall operational efficiency of satellite networks. Network function virtualization (NFV) is viewed as a new paradigm in allocating network resources on-demand. Satellite edge computing combined with the NFV technology is becoming an emerging topic. In this paper, we propose a potential game approach for virtual network function (VNF) placement in satellite edge computing. The VNF placement problem aims to maximize the number of allocated IoT users, while minimizing the overall deployment cost. We formulate the VNF placement problem with maximum network payoff as a potential game and analyze the problem by a game-theoretical approach. We implement a decentralized resource allocation algorithm based on a potential game (PGRA) to tackle the VNF placement problem by finding a Nash equilibrium. Finally, we conduct the experiments to evaluate the performance of the proposed PGRA algorithm. The simulation results show that the proposed PGRA algorithm can effectively address the VNF placement problem in satellite edge computing.
\end{abstract}


%
\IEEEpeerreviewmaketitle

\section{Introduction}
%
%
%
%
\IEEEPARstart{W}{ith} the rapid development of the Internet of Things (IoT) and edge computing technologies, IoT users can be distributed in order to provide wide coverage services in remote areas, e.g., environment monitoring, ocean transportation, smart grid, etc., \cite{7289337}. Considering that IoT users have low latency requirements, limited computing capabilities and battery power, computation tasks from IoT users can be offloaded to nearby edge servers for further performing, where edge servers are usually deployed at base stations (BSs) \cite{8030322}. However, terrestrial networks have not been established in some remote areas of deserts, oceans, and mountains, due to high network construction costs and specific geographical conditions \cite{9048610}. Therefore, it is hard to offer data collection and computation offloading for IoT users only by terrestrial networks in these remote areas. As a supplement to terrestrial networks, low earth orbit (LEO) satellite networks, which have global seamless coverage and low transmission delay time, play an important role in satellite-based IoT and edge computing \cite{8681409,8610431,8002583}.\par

For some remote areas without the coverage of terrestrial networks, LEO satellite networks can assist in gathering data from remote IoT users and transmitting them to cloud data centers on the ground for further processing \cite{8681409}. Due to the nature of LEO satellite networks, the transmission delay between remote IoT users and cloud data centers will be difficult to meet the real-time requirements of IoT users. Besides, the available network bandwidths will decrease to result in the network congestion as the number of IoT users increases. Considering the limited network bandwidths and real-time requirements, we can deploy edge servers on LEO satellites and provide edge computing services for remote IoT users to reduce their end-to-end delay \cite{8610431}, such as ocean transportation and smart grid \cite{7289337}. However, LEO satellites have limited resource capacities of computing, storage, bandwidth, and energy \cite{8896440}. In order to improve the operational efficiency of LEO satellite networks, multiple LEO satellites can provide computing services by the network function virtualization (NFV) technology \cite{7534741} for a IoT user in a cooperative manner.\par
\begin{figure}[tbp]
  \centering
  \includegraphics[width = 0.9\columnwidth]{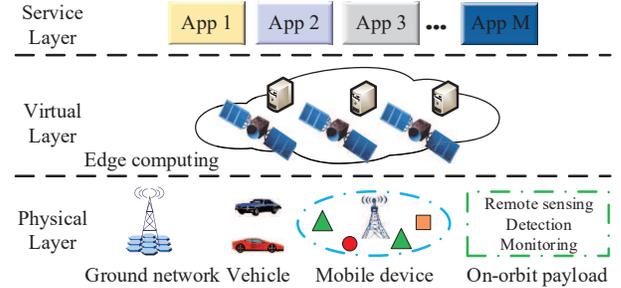}
  \caption{Satellite edge computing framework with NFV.}
  \label{Satellite edge computing framework with NFV}
\end{figure}
As a new paradigm in allocating network resources on-demand, NFV can support the decoupling of software and hardware equipments and enable service functions to run on commodity servers \cite{7534741}. By introducing NFV to satellite edge computing, we can abstract the available resources of LEO satellite networks into a resource pool and provide agile service provisioning for IoT users on-demand \cite{9062531}. Fig.~\ref{Satellite edge computing framework with NFV} shows the satellite edge computing framework with NFV, which consists of physical layer, virtual layer, and service layer. Physical layer consists of remote IoT sensors, actuators, ground networks, etc., and can provide sensing data collection and actuator interaction. For virtual layer, the available resources of LEO satellite networks, e.g., computing, storage, bandwidth, etc., can be abstracted into a resource pool by NFV for allocating available resources to IoT users in a flexible and scalable way. Service layer is responsible for managing LEO satellite network resources and orchestrating virtual network functions (VNFs) for IoT users.\par

For computation offloading in satellite edge computing, most of the existing work focuses on addressing the problem of resource allocation without considering the cooperative operation of multiple LEO satellites, where a computation task from a IoT user is considered as a whole for allocating network resources \cite{wang2019satellite,8672604}. However, these resource allocation strategies can not fully utilize the limited resources of LEO satellite networks. According to service function chaining (SFC), we can decompose a computation task from a IoT user into multiple VNFs in order and deploy these VNFs to multiple LEO satellites correspondingly to improve the network operational efficiency \cite{8690972,8406159}.\par

In this paper, we investigate the VNF placement problem in satellite edge computing. When a IoT user needs to offload its computation task to LEO satellites for obtaining computing services, a user request from the IoT user will be first sent to LEO satellite networks in order to obtain an access permission. The user request is considered as an SFC, consists of multiple VNFs in order, and carries the information concerning service type and resource requirements. Our aim is to maximize the number of allocated IoT users with minimum overall deployment cost, which is composed of energy consumption, bandwidth, and service delay costs. Considering that the user payoff is inversely proportional to the deployment cost, we establish the VNF placement problem with maximum network payoff, which is the sum of all allocated user payoffs.\par

To address the optimization problem of VNF placement, we formulate the problem as a potential game \cite{1996Potential}, which can be performed for making decisions in distributed computing as a non-cooperative game theory and widely used for handling the resource allocation problem in decentralized optimization algorithms \cite{7102682,8823046}. In potential game, a user request from a IoT user is considered as a player for finding a VNF placement strategy with maximum user payoff in a self-interested way and these players have potential conflicts in maximizing their payoffs \cite{8823046}. The payoff for each player can be improved by competing available resources with other players and then a Nash equilibrium can be acquired in a gradual iteration \cite{1996Potential}. Therefore, we implement a decentralized resource allocation algorithm based on a potential game, called as PGRA, to optimize the VNF placement solution. In each iteration, we traverse all the available paths for a user request to find a feasible strategy with maximum user payoff, where the Viterbi algorithm \cite{7469866} is used to address the VNF placement problem for each path. We assume that the resource allocation strategies for these players can be shared by a message synchronization mechanism. Our main contributions of this paper are summarized as follows:
\begin{itemize}
  \item In the perspective of LEO satellite networks, we build the VNF placement problem with maximum network payoff, which is an integer non-linear programming problem. Our aim is to maximize the number of allocated IoT users while minimizing the overall deployment cost including energy consumption, bandwidth, and service delay costs.
  \item To address the VNF placement problem, we formulate the problem as a potential game and analyze the problem by a game-theoretical approach. We implement a decentralized resource allocation algorithm based on a potential game to obtain an approximate strategy profile by finding a Nash equilibrium, where the Viterbi algorithm is used to deploy the VNFs for each user request.
  \item We conduct the experiments to simulate and evaluate the performance of the proposed PGRA algorithm in LEO satellite networks. The simulation results show that the proposed PGRA algorithm outperforms two existing baseline algorithms of Greedy and Viterbi.
\end{itemize}

The remainder of this paper is organized as follows. Section \ref{Related Work} briefly reviews related work about resource allocation in satellite edge computing and decentralized algorithms. Section \ref{System Model} introduces the system model of VNF placement in satellite edge computing. In Section \ref{Problem Formulation}, we model the problem of VNF placement with maximum network payoff and prove it to be NP-hard. The VNF placement problem is formulated as a potential game and a decentralized resource allocation algorithm is implemented for tackling the problem in Section \ref{VNF Placement Game and Proposed Algorithm}. Section \ref{Performance Evaluation} discusses the performance of the proposed PGRA algorithm in LEO satellite networks. Finally, we provide the conclusion of this paper in Section \ref{Conclusion}.\par

\section{Related Work}\label{Related Work}
In satellite edge computing, most of the existing work focuses on offloading computation tasks from IoT devices to satellite edge nodes \cite{wang2019satellite,denby2019orbital,9013467}. In \cite{wang2019satellite}, considering traditional satellites as space edge computing nodes, the authors presented an approach of satellite edge computing to share on-board resources for IoT devices and provided computing services combined with the cloud. The orbital edge computing for nano-satellite constellations was discussed by formulation flying in \cite{denby2019orbital}. A fine-grained resource management in satellite edge computing was presented by the advanced K-means algorithm in \cite{9013467}. This existing literature considers a computation task as a whole to allocate the network resources. In this paper, we assume that each user request is viewed as an SFC and consists of multiple VNFs in order. We need to deploy these VNFs to LEO satellites and route traffic flows between two adjacent VNFs by inter-satellite links.\par

In order to effectively utilize the limited network resources, the VNF orchestration in software defined satellite networks has been investigated in \cite{8690972,8406159,9062531}. The authors in \cite{9062531} formulated the SFC planning problem as an integer non-linear programming problem and proposed a greedy algorithm to address it. The authors in \cite{8690972} discussed the problem of VNF placement to minimize the cost in software defined satellite networks and proposed a time-slot decoupled heuristic algorithm to solve this problem. In \cite{8406159}, an approach of deploying SFCs in satellite networks was presented to minimize the end-to-end service delay. In this paper, we jointly consider three deployment costs of network energy, network bandwidth, and user service delay to formulate the VNF placement problem with maximum network payoff.\par

As the number of remote IoT devices increases, decentralized mechanisms of network resource management have been a research topic in distributed networks \cite{8519737,8360511,8327930}, where potential game is widely used to address the problem of resource allocation in distributed computing \cite{7102682,8823046,8945402,8039433}. The authors in \cite{7102682} proposed a game-theoretical algorithm to minimize the number and power of anchor nodes in a wireless sensor network. In \cite{8823046}, a cost-effective edge user allocation (EUA) problem in edge computing was presented to maximize the number of served users with minimum system cost and the authors designed a decentralized algorithm by a potential game to address the EUA problem. Furthermore, computation offloading in satellite edge computing was discussed by a game-theoretical approach in \cite{8945402}. However, in this paper, we formulate the VNF placement problem with maximum network payoff as a potential game and use a decentralized resource allocation algorithm to optimize the problem.\par

\section{System Model}\label{System Model}
In this section, we discuss the system model of VNF placement in satellite edge computing \cite{wei2019satellite,9048610}, including physical network, user requests, and the VNF placement problem.\par

\subsection{Satellite Network}

We denote a satellite network as a directed graph $G(V,E)$. The parameter $V=\left \{v_n|n=1,2,\cdots,N \right\}$ indicates the set of satellite nodes, where the number of satellite nodes is $N$. We assume that the set of resource types supported by satellite node $v_n$ is denoted by $R_n$ and each satellite node has limited resource capacities, where two resource types of central processing unit (CPU) and storage are considered in this paper. Let us denote the $r$-th resource capacity of satellite node $v_n$ by $C_{n}^{r}$, $r \in R_n$. The parameter $E$ indicates the set of links between satellite nodes. We assume that each satellite has four inter-satellite links (ISLs) with neighbouring satellites. For link $e$, $e \in E$, we denote the bandwidth capacity by $B_{e}$ and the transmission delay time by $t_{e}$. The parameter $P_{n_{1}}^{n_{2}}$ indicates the set of the $d$ shortest paths between satellites $v_{n_1}$ and $v_{n_2}$.\par

Due to the limited power of satellites, one of our aims is to minimize the overall energy cost of a satellite network. We introduce a power consumption model for edge servers on satellite nodes \cite{7279063}. An edge server can be considered in four states of \emph{on}, \emph{idle}, \emph{unavailable off}, and \emph{available off}. In an \emph{on} state, an edge server can provide computing services for IoT users and will consume energy. We denote $P_{n}^{on}$ as the average power consumption of an edge server on satellite node $v_n$ in an \emph{on} state. If an edge server on satellite node $v_n$ does not provide computing services for any IoT users the edge server is considered into an \emph{idle} state and we use $P_{n}^{idle}$ to indicate the average idle power consumption. When the idle time for an edge server on satellite node $v_n$ is over the maximum idle threshold $t_{n}^{idle}$ the edge server will be into an \emph{unavailable off} state. If an edge server is in an \emph{unavailable off} state, it can not provide computing services for IoT users in the next time slot. When the off time for an edge server on satellite node $v_n$ is greater than the minimum off threshold $t_{n}^{off}$ the edge server can be in an \emph{available off} state, that is, the edge server can provide computing services for IoT users in the next time slot. If an edge server in an \emph{available off} state needs to provide computing services for IoT users there will exist a setup procedure for the edge server, where the state of the edge server will be converted from \emph{off} to \emph{on}. We assume that the period of the setup procedure is $1$ time slot and the average setup power consumption is considered as the maximum power $P_{n}^{max}$ of an edge server on satellite node $v_n$. For \emph{available} and \emph{unavailable} states, the power consumption of an edge server is zero. Note that a satellite node can be considered as a router for routing traffic flows when its edge server is in an \emph{off} state. Based on the above discussion, the state transition diagram of an edge server is shown in Fig.~\ref{State transition diagram of an edge server}.\par
\begin{figure}[tbp]
  \centering
  \includegraphics[width = 0.9\columnwidth]{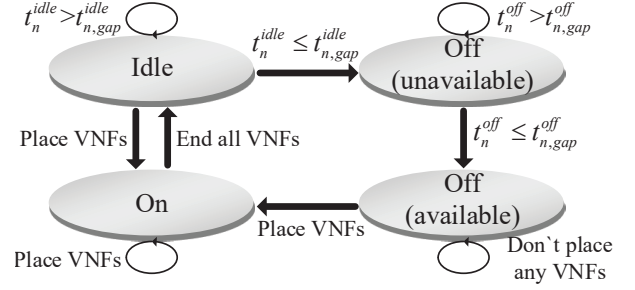}
  \caption{State transition diagram of an edge server.}
  \label{State transition diagram of an edge server}
\end{figure}
\subsection{User Requests}
When a IoT user needs to offload its computation task to satellites, a user request will be first sent to the LEO satellite network for obtaining an access permission, where the user request consists of multiple VNFs in specific order and can be considered as an SFC. We denote a set of user requests by $U=\left \{u_m|m=1,2,\cdots,M \right\}$ with $M$ user requests. The user request $u_{m}$, $u_{m} \in U$ is defined as a directed graph $G(F_{m},H_{m})$. The set $F_{m}=\left \{f_{m,1}=s_m,f_{m,2},\cdots,f_{m,\left |F_{m} \right |}=d_m \right\}$ indicates the set of VNFs for user request $u_m$, where $f_{m,i}$ indicates the $i$-th VNF of user request $u_m$, while the parameters $s_m$ and $d_m$ indicate the source and the destination of user request $u_m$, respectively. In satellite edge computing, we assume that the results of computation tasks processed by satellite nodes can be sent back to IoT users or transmitted to cloud data centers on ground. In these scenarios, the source and the destination of a user request can either be the same node or not. The resource requirements of each VNF include computing, storage, and execution time. We use $c_{m,i}^{r}$ to indicate the $r$-th resource requirement of $f_{m,i}$ and $t_{m,i}$ to represent the execution time of $f_{m,i}$. The set $H_{m}$ describes the set of edges for user request $u_m$. An edge between $f_{m,i_1}$ and $f_{m,i_2}$ is denoted by $h_{m}^{i_{1},i_{2}}$ and the bandwidth requirement of edge $h_{m}^{i_{1},i_{2}}$ is denoted by $b_{m}^{i_{1},i_{2}}$ accordingly. We define the maximum acceptable delay time for user request $u_m$ as $t_{m}^{delay}$.\par

\subsection{VNF Placement Problem}
\begin{figure}[tbp]
  \centering
  \includegraphics[width = 0.9\columnwidth]{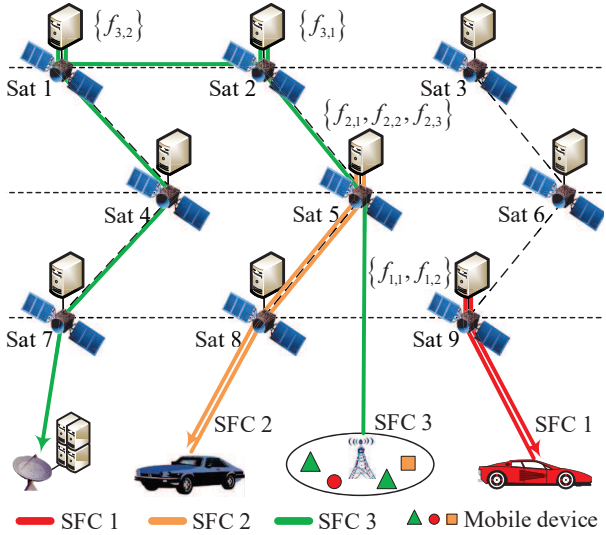}
  \caption{Example of placing VNFs for three IoT users.}
  \label{Example of placing VNFs for three IoT users}
\end{figure}
In this paper, we discuss the problem of dynamically deploying VNFs over varying time slots. A batch processing mode is simply used for allocating available resources of a LEO satellite network to user requests. We assume that a batch of user requests arrive concurrently at the beginning of each time slot and they can make decisions for deploying the VNFs to satellite nodes by their optimization strategies, where user requests have different resource requirements and maximum acceptable delay time, however, the source and the destination for each user request can be known. In addition, we assume that the satellite network topology is constant during the computing service period for each user request. Thus, the variability of a satellite network does not need to be considered when we deploy VNFs for user requests to satellite nodes. For deploying VNFs in dynamic cloud environment, we first abstract all available resources of a satellite network into a resource pool. Before performing resource allocation algorithms for the VNF placement, we need to free the resources that are used by the completed user requests in the previous time slot into the resource pool as available resources for new user requests in the current time slot. Then under remaining available resource and service requirement constraints, we can run resource allocation algorithms to orchestrate VNFs for new user requests with maximum network payoff. Our aim is to provide computing services for as many user requests as possible with minimum overall deployment cost, which consists of energy, bandwidth, and service delay costs.\par

For a user request, the computation task can be offloaded to satellites by different resource allocation strategies. However, the VNF placement strategy has an impact on the network payoff. Fig.~\ref{Example of placing VNFs for three IoT users} shows an example of placing VNFs for three IoT users. There are 9 satellite nodes, which are represented by $\left \{Sat1,Sat2,\cdots,Sat9\right\}$, and each satellite node deploys an edge server. The number of inter-satellite links for a satellite is 4. The user requests are composed of multiple specific VNFs in sequence. The three user requests can be described by $SFC1=\left \{s_1,f_{1,1},f_{1,2},d_1\right\}$, $SFC2=\left \{s_2,f_{2,1},f_{2,2},f_{2,3},d_2\right\}$, and $SFC3=\left \{s_3,f_{3,1},f_{3,2},d_3\right\}$, respectively. We assume that source $s_1$ and destination $d_1$ for $SFC1$ are on the same satellite node $Sat9$, source $s_2$ and destination $d_2$ for $SFC2$ are also on the same satellite node $Sat8$. In addition, source $s_3$ and destination $d_3$ for $SFC3$ are on satellite nodes $Sat4$ and $Sat7$, respectively. We deploy VNFs $f_{1,1}$ and $f_{1,2}$ for $SFC1$ to satellite node $Sat9$, VNFs $f_{2,1},f_{2,2}$, and $f_{2,3}$ for $SFC2$ to satellite node $Sat5$. For $SFC3$, we deploy VNF $f_{3,1}$ to satellite node $Sat2$ and VNF $f_{3,2}$ to satellite node $Sat1$, respectively. For $SFC1$, all VNFs are deployed on satellite node $Sat9$. The routing path for $SFC2$ is $\left \{Sat8,Sat5,Sat8\right\}$ and for $SFC3$ is $\left \{Sat5,Sat2,Sat1,Sat4,Sat7\right\}$. The strategies for the three IoT users are shown in Table \ref{VNF Placement Solutions for Three IoT Users}.\par
\begin{table}[tbp]
  \centering
  \caption{VNF Placement Solutions for Three IoT Users.}
  \label{VNF Placement Solutions for Three IoT Users}%
  \resizebox{\columnwidth}{!}{
    \begin{tabular}{cc}
    \hline
    User & Strategy \\
    \hline
    SFC1  & $s_1$(Sat9)$\rightarrow$ $f_{1,1}$(Sat9) $\rightarrow$ $f_{1,2}$(Sat9)$\rightarrow$ $d_1$(Sat9) \\
    \hline
    SFC2  & $s_2$(Sat8) $\rightarrow$ $f_{2,1}$(Sat5) $\rightarrow$ $f_{2,2}$(Sat5) $\rightarrow$ $f_{2,3}$(Sat5) $\rightarrow$ $d_2$(Sat8) \\

    \hline
    SFC3  & $s_3$(Sat5) $\rightarrow$ $f_{3,1}$(Sat2) $\rightarrow$ $f_{3,2}$(Sat1) $\rightarrow$ Sat4 $\rightarrow$ $d_3$(Sat7)\\
    \hline
    \end{tabular}%
  }
\end{table}%
\section{Problem Formulation}\label{Problem Formulation}
In this section, the problem of VNF placement is proposed by a mathematical method in satellite edge computing and proved to be NP-hard.
\subsection{Problem Description}
\begin{table}[tbp]
  \caption{List of Symbols.}
  \label{List of Symbols}
  \centering
  \resizebox{\columnwidth}{!}{
  \begin{tabular}{cl}
  \hline
  \multicolumn{2}{c}{\bfseries Satellite Network}\\
  \hline
  $V$ & Set of satellites with the number of $N$.\\
  $v_n$ & The $n$-th satellite.\\
  $R_n$ & Set of resources offered by satellite $v_n$.\\
  $C_{n}^{r}$ & The $r$-th resource capacity for satellite $v_n$.\\
  $P_{n}^{idle}$ & Idle power of an edge server on satellite $v_n$.\\
  $P_{n}^{on}$ & Active power of an edge server on satellite $v_n$.\\
  $P_{n}^{max}$ & Maximum power of an edge server on satellite $v_n$.\\
  $t_{n}^{idle}$ & Maximum idle time of an edge server on satellite $v_n$.\\
  $t_{n}^{off}$ & Minimum off time of an edge server on satellite $v_n$.\\
  $E$ & Set of links between satellites.\\
  $e$ & The $e$-th link.\\
  $B_{e}$ & Bandwidth capacity for link $e$.\\
  $t_{e}$ & Transmission delay for link $e$.\\
  $P_{n_{1}}^{n_{2}}$ & Set of the $d$ shortest paths between $v_{n_{1}}$ and $v_{n_{2}}$.\\
  \hline
  \multicolumn{2}{c}{\bfseries User Requests}\\
  \hline
  $U$ & Set of $M$ user requests.\\
  $u_{m}$ & The $m$-th user request.\\
  $F_{m}$ & Set of VNFs for user request $u_m$.\\
  $f_{m,i}$ & The $i$-th VNF for $u_m$.\\
  $s_m,d_m$ & Source and destination of user request $u_m$.\\
  $c_{m,i}^{r}$ & The $r$-th resource requirement for VNF $f_{m,i}$.\\
  $t_{m,i}$ & Execute time for VNF $f_{m,i}$.\\
  $H_{m}$ & Set of edges for user request $u_m$.\\
  $h_{m}^{i_{1},i_{2}}$ & Edge between VNFs $f_{m,i_{1}}$ and $f_{m,i_{2}}$.\\
  $b_{m}^{i_{1},i_{2}}$ & Bandwidth resource requirement for $h_{m}^{i_{1},i_{2}}$. \\
  $t_{m}^{delay}$ & Maximum acceptable delay time for user request $u_m$. \\
  \hline
  \multicolumn{2}{c}{\bfseries Binary Decision Variables}\\
  \hline
  $x_{m,i}^{n}$ & $x_{m,i}^{n} = 1$ if $f_{m,i}$ is placed on satellite $v_n$ or $x_{m,i}^{n} = 0$.\\
  $y_{m,p}^{i_{1},i_{2}}$ & $y_{m,p}^{i_{1},i_{2}}=1$ if path $p$ is used by $h_{m}^{i_{1},i_{2}}$ or $y_{m,p}^{i_{1},i_{2}}=0$.\\
  \hline
  \multicolumn{2}{c}{\bfseries Variables}\\
  \hline
  $q_{e}^{p}$ & $q_{e}^{p}=1$ if link $e$ is used by path $p$, otherwise $q_{e}^{p}=0$.\\
  $z_m$ & $z_m=1$ if $u_m$ is allocated, otherwise $z_m=0$.\\
  $\varphi_m^{bw}$ & Bandwidth cost for user request $u_m$.\\
  $\varphi_m^{power}$ & Energy cost for user request $u_m$.\\
  $\varphi_m^{delay}$ & Delay cost for user request $u_m$.\\
  $\varphi_m$ & Payoff function for user request $u_m$.\\
  $\Phi$ & Total payoff function. \\
  $\alpha$ & Weight value. \\
  \hline
  \end{tabular}
  }
\end{table}
In the perspective of a satellite network, we formulate the VNF placement problem as a constrained optimization problem with maximum network payoff in satellite edge computing, where the VNF placement problem is viewed as an integer non-linear programming problem. We assume that when a satellite network provides computing services for user requests, it can acquire uncertain user payoffs, which are affected by the deployment costs of user requests. Considering the limited physical resources of satellites and the real-time service requirement of IoT users, the goal of VNF placement in satellite edge computing is to reduce the energy consumption and bandwidth usage for a satellite network while minimizing the end-to-end user delay. Therefore, we assume that the deployment cost consists of three parts as: (1) energy cost, (2) bandwidth cost, and (3) service delay cost. The user payoff for a user request is non-negative and inversely proportional to the deployment cost. Therefore, the lower the deployment cost for a user request is, the higher the user payoff will be. When a user request is not deployed to satellite nodes, we denote the user payoff as zero. The overall network payoff is the sum of all user payoffs and our optimization aim is to maximize the network payoff within the network physical resource and service requirement constraints. That is, we make the number of allocated user requests maximum while minimizing the overall deployment cost. To better discuss the problem of VNF placement, we list the main symbols for our problem description in Table \ref{List of Symbols}.\par

Let us denote a binary decision variable $x_{m,i}^{n}=\left \{0,1\right \}$ to indicate whether VNF $f_{m,i}$ is placed on satellite node $v_n$, where $x_{m,i}^{n}=1$ if VNF $f_{m,i}$ is placed on satellite node $v_n$, otherwise $x_{m,i}^{n}=0$.\par

Another binary decision variable $y_{m,p}^{i_{1},i_{2}}=\left \{0,1\right \}$ is defined to describe which path is used by edge $h_{m}^{i_1,i_2}$. If path $p$ is used by $h_{m}^{i_1,i_2}$, then $y_{m,p}^{i_{1},i_{2}}=1$, otherwise $y_{m,p}^{i_{1},i_{2}}=0$.\par

We also denote a binary variable $q_{e}^{p}$ to indicate whether link $e$ is used by path $p$, as:
\begin{equation}\label{equation0}
q_{e}^{p}=
\begin{cases}
1& \text{if link $e$ is used by path $p$}, \\
0& \text{otherwise}.
\end{cases}
\end{equation}
A  binary variable $z_m=\left \{0,1\right \}$ is used to indicate whether user request $u_m$ is deployed to satellite nodes, as:
\begin{equation}\label{equation1}
z_m=
\begin{cases}
1& \text{if user request $u_m$ is deployed to satellites}, \\
0& \text{otherwise}.
\end{cases}
\end{equation}
When we deploy the VNFs for a user request to satellites, the deployment cost can be composed of energy consumption, bandwidth, and service delay costs, where the three costs are normalized values.\par

\noindent The bandwidth resources $u_m^{bw}$ used by user request $u_m$ can be described as:
\begin{equation}\label{equation2}
u_m^{bw} =\!\! \sum\limits_{h_m^{{i_1},{i_2}} \in {H_m}}\!\! {\sum\limits_{{v_{{n_1}}},{v_{{n_2}}}}\! {\sum\limits_{p \in P_{{n_1}}^{{n_2}}} {\sum\limits_{e \in p} {x_{m,{i_1}}^{{n_1}} \!\cdot\! x_{m,{i_2}}^{{n_2}} \!\cdot\! y_{m,p}^{{i_1},{i_2}} \!\cdot\! q_e^p \!\cdot\! b_m^{{i_1},{i_2}}} } } }.
\end{equation}
The total bandwidth resource capacities in satellite network $G(V,E)$ are $\sum\limits_{e \in E}B_e$. Therefore, the normalized bandwidth cost for user request $u_m$ can be described as:
\begin{equation}\label{equation3}
\varphi_m^{bw}=\frac{u_m^{bw}}{\sum\limits_{e \in E}B_e}.
\end{equation}
According to the power consumption model \cite{7279063}, we assume that the energy consumption for an edge server on a satellite is mainly produced by running CPU. The energy consumption for an active edge server on satellite node $v_n$ can be expressed as:
\begin{equation}\label{equation4}
P_n^{on}= P_n^{idle} + \frac{\sum\limits_{u_m \in U}\sum\limits_{f_{m,i} \in F_m}x_{m,i}^{n}\cdot c_{m,i}^{cpu}}{C_{n}^{cpu}}\cdot (P_n^{max} -P_n^{idle}).
\end{equation}
Considering that there are different energy consumptions for placing VNFs on an edge server in various running states, we divide the VNF placement into four solutions based on energy consumption costs as follows:
\begin{itemize}
  \item \emph{Case 1:} An edge server on satellite node $v_n$ is \emph{off} in the current time slot and will provide computing services for VNFs in the next time slot. Then if we place VNF $f_{m,i}$ on satellite node $v_n$, the power produced by VNF $f_{m,i}$ in the current time slot can be denoted as $u_{m,i}^{power}=0$.
  \item \emph{Case 2:} An edge server on satellite node $v_n$ is \emph{off} in the current time slot and will not provide computing services for any VNFs in the next time slot. Then if we place VNF $f_{m,i}$ on satellite node $v_n$, the power produced by VNF $f_{m,i}$ in the current time slot can be denoted as $u_{m,i}^{power}=P_{n}^{max}$.
  \item \emph{Case 3:} An edge server on satellite node $v_n$ is \emph{idle} in the current time slot and will not provide computing services for any VNFs in the next time slot. Then if we place VNF $f_{m,i}$ on satellite node $v_n$, the power produced by VNF $f_{m,i}$ in the current time slot could be denoted as $u_{m,i}^{power}=P_n^{idle} + \frac{x_{m,i}^{n}\cdot c_{m,i}^{cpu}}{C_{n}^{cpu}} \cdot (P_n^{max} -P_n^{idle}).$
  \item \emph{Case 4:} An edge server on satellite node $v_n$ is \emph{idle} or \emph{on} in the current time slot and will provide computing services for VNFs in the next time slot. Then if we place VNF $f_{m,i}$ on satellite node $v_n$, the power produced by VNF $f_{m,i}$ in the current time slot could be seen as $u_{m,i}^{power}=\frac{x_{m,i}^{n}\cdot c_{m,i}^{cpu}}{C_{n}^{cpu}}\cdot (P_n^{max} -P_n^{idle})$.
\end{itemize}
The total energy consumption in satellite network $G(V,E)$ is represented by $\sum\limits_{v_n \in V}P_n^{max}$. The normalized energy cost for user request $u_m$ can be expressed as:
\begin{equation}\label{equation5}
\varphi_m^{power}=\frac{1}{\sum\limits_{v_n \in V}P_n^{max}} \sum\limits_{f_{m,i}\in F_m}\sum\limits_{v_n \in V} x_{m,i}^{n}\cdot u_{m,i}^{power}.
\end{equation}
For a user request, the source-to-destination delay time consists of two parts: computing delay for VNFs and transmission delay. The computing delay time for user request $u_m$ can be denoted by:
\begin{equation}\label{equation6}
t_m^{exec}=\sum_{f_{m,i}\in F_m}\sum\limits_{v_n \in V} x_{m,i}^{n} \cdot t_{m,i}.
\end{equation}
The transmission delay time for user request $u_m$ can be indicated as:
\begin{equation}\label{equation7}
t_m^{trans} =\!\! \sum\limits_{h_m^{{i_1},{i_2}} \in {H_m}}\!\! {\sum\limits_{{v_{{n_1}}},{v_{{n_2}}}}\! {\sum\limits_{p \in P_{{n_1}}^{{n_2}}} {\sum\limits_{e \in p} {x_{m,{i_1}}^{{n_1}} \!\cdot\! x_{m,{i_2}}^{{n_2}} \!\cdot\! y_{m,p}^{{i_1},{i_2}} \!\cdot\! q_e^p \!\cdot\! t_e} } } }.
\end{equation}
The maximum acceptable delay time for user request $u_m$ is $t_{m}^{delay}$, then we can denote the normalized service delay time cost for user request $u_m$ by:
\begin{equation}\label{equation8}
\varphi_m^{delay}=\frac{1}{t_{m}^{delay}}\left ( t_m^{exec}+ t_m^{trans}\right ).
\end{equation}
For user request $u_m$, the deployment cost is the weighted sum of energy consumption cost $\varphi_m^{power}$, bandwidth cost $\varphi_m^{bw}$, and service delay cost $\varphi_m^{delay}$, then the user payoff $\varphi_m$ can be indicated by:
\begin{equation}\label{equation9}
\varphi_m= (1 - \alpha_{1} \cdot \varphi_m^{bw} - \alpha_{2} \cdot \varphi_m^{power} - \alpha_{3} \cdot \varphi_m^{delay} ) \cdot z_m,
\end{equation}
where $\alpha_{1}$, $\alpha_{2}$, and $\alpha_{3}$, $\alpha_{1}+\alpha_{2}+\alpha_{3}=1$, are the weighted factors and can adjust the preferences of the three costs.\par
\noindent The overall network payoff $\Phi$ is the sum of all user payoffs and can be represented by:
\begin{equation}\label{equation10}
\Phi= \sum\limits_{u_m \in U} \varphi_m.
\end{equation}

In order to address the problem of VNF placement to maximize the overall network payoff, the following physical constraints need to be considered.\par
\noindent When the VNFs for user request $u_m$ are deployed to satellite nodes, each VNF $f_{m,i} \in F_m$ can be placed on one and only one satellite node. We describe the VNF deployment constraint as follows:
\begin{equation}\label{equation11}
\sum\limits_{v_n \in V} x_{m,i}^{n} = 1,\forall f_{m,i}\in F_m,\forall u_m \in U.
\end{equation}
For user request $u_m$, if two adjacent VNFs $f_{m,i_1}$ and $f_{m,i_2}$ are deployed on satellite nodes $v_{n_1}$ and $v_{n_2}$, respectively, we need to guarantee that there exists one path between $v_{n_1}$ and $v_{n_2}$ that can be used to route traffic flows from $f_{m,i_1}$ to $f_{m,i_2}$. The path selection constraint can be expressed by:
\begin{equation}\label{equation12}
\sum\limits_{p \in P_{n_1}^{n_2}}\!\!\! y_{m,p}^{i_1,i_2}\! = \!x_{m,i_1}^{n_1}\!\!\cdot \!\!x_{m,i_2}^{n_2}, f_{m,i_1},f_{m,i_2}\!\in \!F_m, v_{n_1},v_{n_2}\! \in \!V.
\end{equation}
When we place VNFs for user requests to satellite nodes, the resource requirements for each satellite node can not exceed its resource capacities. The satellite resource constraint can be described as:
\begin{equation}\label{equation13}
\sum\limits_{u_m \in U}\sum\limits_{f_{m,i}\in F_m} x_{m,i}^{n}\cdot c_{m,i}^{r}\leq C_{n}^{r},\forall v_{n}\in V,\forall r \in R_n.
\end{equation}
When we choose a path between two satellite nodes to route traffic flows, we also guarantee that the bandwidth requirements for each link can not be greater than the bandwidth capacity. The bandwidth resource constraint for $\forall e \in E$ can be indicated by:
\begin{equation}\label{equation14}
\sum\limits_{u_m \in U}\! {\sum\limits_{h_m^{{i_1},{i_2}} \in {H_m}}\!\! {\sum\limits_{{v_{{n_1}}},{v_{{n_2}}}}\! {\sum\limits_{p \in P_{{n_1}}^{{n_2}}} {\!\!\!\!{x_{m,{i_1}}^{{n_1}} \!\!\cdot\! x_{m,{i_2}}^{{n_2}} \!\cdot\! y_{m,p}^{{i_1},{i_2}} \!\cdot\! q_e^p \!\!\cdot\! b_m^{{i_1},{i_2}}} } } }}\!\!\leq \!\!B_{e}.
\end{equation}
For user request $u_m$, the source-to-destination delay time is less than the maximum acceptable delay time. The service delay time constraint can be described as:
\begin{equation}\label{equation15}
t_m^{exec}+ t_m^{trans}\leq t_{m}^{delay},\forall u_m \in U,
\end{equation}
where the maximum acceptable delay time for a user request is the sum of the executed time for all VNFs and the acceptable path transmission time. In this paper, we assume that the acceptable path transmission time is equal to the average transmission time of all source-to-destination paths $P_{s_m,all}^{d_m}$ in satellite network $G(V,E)$. Thus, we can denote the maximum acceptable delay time for user request $u_m$ by:
\begin{equation}\label{equation16}
t_{m}^{delay} = \sum_{f_{m,i}\in F_m}t_{m,i} + \frac{1}{\left |P_{s_m,all}^{d_m}\right |}\sum\limits_{p\in P_{s_m,all}^{d_m}}\sum\limits_{e\in p} t_{e}.
\end{equation}
In addition, we need to ensure that the idle time $t_{n,gap}^{idle}$ for an edge server on satellite $v_n$ can not exceed the maximum idle time threshold $t_{n}^{idle}$. For an edge server on satellite $v_n$, we denote the earliest idle time in the current \emph{idle} state by $t_{n,idle}^{old}$ and the current idle time by $t_{n,idle}^{new}$. The idle time constraint for an edge server on satellite $v_n$ can be indicated by:
\begin{equation}\label{equation17}
t_{n,gap}^{idle}=t_{n,idle}^{new} - t_{n,idle}^{old}\leq t_{n}^{idle},\forall v_n \in V.
\end{equation}
We also guarantee that the off time $t_{n,gap}^{off}$ for an edge server on satellite $v_n$ should be greater than the minimum off time threshold $t_{n}^{off}$. For an edge server on satellite $v_n$, we use $t_{n,off}^{old}$ to indicate the earliest off time in the current \emph{off} state and $t_{n,off}^{new}$ to indicate the current off time. The off time constraint for an edge server on satellite $v_n$ can be indicated by:
\begin{equation}\label{equation18}
t_{n,gap}^{off}=t_{n,off}^{new} - t_{n,off}^{old} \geq t_{n}^{off},\forall v_n \in V.
\end{equation}
Under the physical network resource and service requirement constraints in equations \eqref{equation11}-\eqref{equation18}, we formulate the VNF placement problem in satellite edge computing using \eqref{equation10} as an optimization problem with maximum network payoff. The optimization problem in this paper can be indicated by:
\begin{equation}\label{equation19}
\begin{aligned}
\text{max}\quad & \Phi \\
s.t.\quad &  \eqref{equation11}-\eqref{equation18}.
\end{aligned}
\end{equation}
\subsection{Problem Analysis}
In this section, we reduce the capacitated plant location problem with single source constraints (CPLPSS) \cite{1995The} to the above VNF placement problem and prove that the formulated VNF placement problem is NP-hard \cite{7469866}.

In CPLPSS, a set of potential locations is given for deploying plants with fixed capacities and costs, and goods from the plants need to be provided to a set of customers with fixed demands. There are transportation costs for supplying goods from the plants to the customers. In addition, the goods demanded by a customer are provided only by a single plant. The problem aims to find an optimal solution of placing plants within the capacity and demand constraints to minimize the total operational and transportation costs.

To reduce a CPLPSS problem to our proposed VNF placement, we re-construct the problem description. Satellite nodes indicate plants and satellite resources represent plant capacities. User requests can be viewed as customers and resource demands for their computation tasks are described as customer required goods. We assume that all demanded resources from a user request are offered only by a satellite node. The operational cost of a plant is denoted by energy consumption and the transportation cost can be indicated by bandwidth and source-to-destination delay costs. For the proposed VNF placement problem, maximum overall network payoff is equal to minimizing the sum cost of energy consumption, bandwidth, and service delay costs within the constant number of user requests. Thus, a CPLPSS problem can be reduced to the proposed optimization problem. As a CPLPSS problem is NP-hard, the VNF placement problem is also NP-hard.

\section{VNF Placement Game and Proposed Algorithm}\label{VNF Placement Game and Proposed Algorithm}
In this section, we formulate the VNF placement problem as a potential game and analyze its property by a game-theoretical approach. Then a decentralized resource allocation algorithm based on a potential game is proposed for addressing the VNF placement problem.

\subsection{VNF Placement Game}
Game theory is a mathematical tool for analyzing interactive decision-making processes. A non-cooperative game can be denoted by $\Gamma =\left \{U,A,\varphi\right \}$. The term $U$ indicates the set of players. The strategy of the $m$-th player is denoted by $a_m \in A_m$, where $A_m$ represents the strategy set of the $m$-th player. $A=\prod\limits_{m = 1}^M {{A_m}}$ denotes the strategy combination of all players and $a=\left \{a_1,a_2,\cdots,a_M\right \}$ indicates the strategies of all players. We denote the strategies of all players except the $m$-th player as $a_{-m}=\left \{a_1,\cdots,a_{m-1},a_{m+1},\cdots,a_M\right \}$ and the strategy combination of all players except the $m$-th player as $A_{-m}=\prod\limits_{j \neq m}^M {{A_j}}$. The term $\varphi_{m}(a_m) \in \varphi$ indicates the payoff of the $m$-th player with the strategy $a_m$ and the term $\varphi(a)=\left \{\varphi_{1}(a_1),\varphi_{2}(a_2),\cdots,\varphi_{M}(a_M)\right \}$ is the payoff set of all players.

For a non-cooperative game, the players can make strategy decisions in a self-interested way for maximizing their payoffs until a Nash equilibrium is generated, where each player can not unilaterally deviate its strategy for improving the payoff.
\begin{definition}\label{definition1}
(\textbf{Nash equilibrium}) A strategy decision profile $a^*=\left \{a_1^*,a_2^*,\cdots,a_M^*\right \}$ of all user requests is a Nash equilibrium if no player has an incentive for unilaterally deviating the strategy, such that,
\begin{equation}\label{equation20}
\varphi_{m}(a_m^*,a_{-m}^*)\ge \varphi_{m}(a_m,a_{-m}^*), \forall u_m \in U, \forall a_m \in A_m.
\end{equation}
\end{definition}
Before finding a Nash equilibrium of the game, we should ensure whether the game admits at least a Nash equilibrium. As a special instance of non-cooperative games, potential game possesses a pure strategy Nash equilibrium \cite{1996Potential}, where a global potential function is used to map the payoffs of all players. The game can be considered as an exact potential game if an exact potential function is admitted.
\begin{definition}\label{definition2}
(\textbf{Exact Potential Game}) A game is an exact potential game if, for an exact potential function $\Phi(a), \forall u_m \in U, a_m, a_m' \in A_m$ and $a_{-m} \in A_{-m}$, there is,
\begin{equation}\label{equation21}
\Phi(a_m',a_{-m})\! - \!\Phi(a_m,a_{-m})\! = \!\varphi_{m}(a_m',a_{-m})\! - \!\varphi_{m}(a_m,a_{-m}).
\end{equation}
\end{definition}
The key for the existence of a Nash equilibrium is to prove the VNF placement game is a potential game. For the VNF placement in satellite edge computing, user requests have potential conflicts in maximizing their payoffs and the strategy of a user request has an effect on that of other user requests. We formulate the VNF placement problem as a non-cooperative game, where $M$ user requests are $M$ players, the VNF placement solution for user request $u_m$ represents the decision strategy $a_m$ and the payoff for user request $u_m$ indicates the payoff of the $m$-th player. The term $A_m$ indicates the strategy set of user request $u_m$ and $A$ is the strategy combination of all user requests. The exact potential function is denoted by $\Phi(a)$, which is the sum of all user request payoffs. Then we prove that the VNF placement game is a potential game.
\begin{proposition}\label{proposition1}
The VNF placement game is an exact potential game, where the payoff of the $m$-th player is indicated by $\varphi_{m}(a_m)$ and the exact potential function is the network payoff $\Phi(a)$.
\end{proposition}
\begin{proof}\label{proof1}
For $a_m',a_m \in A_m$, $a_{-m} \in A_{-m}$, according to equation \eqref{equation9}, there is,
\begin{equation}\label{equation22}
\Delta\varphi_{m} = \varphi_{m}(a_m',a_{-m}) - \varphi_{m}(a_m,a_{-m}).
\end{equation}
Similarly, according to equation \eqref{equation10}, we can obtain the potential function difference as:
\begin{equation}\label{equation23}
\begin{aligned}
\Delta\Phi\! & = \!\Phi(a_m',a_{-m}) - \Phi(a_m,a_{-m}) \\
 & \! = \!\!\left [\!\varphi_{m}(a_m',a_{-m}) \!+ \!\!\!\!\!\sum\limits_{u_l\neq u_m}^U \!\!\!\!\varphi_l(a_l)\! \right ] \!\!\!-\!\!\! \left [ \!\varphi_{m}(a_m,a_{-m}) \!+ \!\!\!\!\!\sum\limits_{u_l\neq u_m}^U \!\!\!\!\varphi_l(a_l)\!\right ]\\
 & \! = \!\varphi_{m}(a_m',a_{-m}) - \varphi_{m}(a_m,a_{-m}).
\end{aligned}
\end{equation}
There is $\Delta\Phi \equiv \Delta\varphi_{m}$ \cite{2015ChoiA,7102682}. Thus, we prove that the VNF placement game is an exact potential game and $\Phi(a)$ is an exact potential function.
\end{proof}
Considering that the VNF placement game is a potential game and has the finite improvement property \cite{2018A}, the VNF placement problem can be addressed by finding a Nash equilibrium in a finite gradual iteration. A decentralized resource allocation algorithm based on a potential game, which is discussed in the following subsection, is implemented to tackle the VNF placement problem.

\subsection{Decentralized Resource Allocation Algorithm}
\begin{algorithm}[tbp]
  \caption{Resource Allocation Based on a Potential Game.}
  \label{Resource Allocation Based on a Potential Game}
  \hspace*{0.02in} {\bf Input:} User requests $U$;\\
  \hspace*{0.02in} {\bf Output:} $a^{*}=\left[a_{1}^{*},a_{2}^{*},\cdots, a_{M}^{*}\right]$;

  \begin{algorithmic}[1]
  \STATE \textbf{Initialize:} $k=0,\forall u_m \in U,a_m(k)=\varnothing$;
  \WHILE {$k<K_{max}$}
  \FOR {$\forall u_m \in U$}
  \FOR {$\forall p \in P_{s_m}^{d_m}$}
  \STATE Search an optimal strategy $a'_m(k,p)$ for path $p$ by the Viterbi algorithm and compute the user payoff $\varphi (a'_m(k,p),a_{-m}(k))$;
  \ENDFOR
  \STATE Find the strategy $a'_m(k)$ with maximum user payoff according to\\ $a'_m(k) = \arg\max\limits_{a'_m(k,p)} \varphi (a'_m(k,p),a_{-m}(k))$;
  \IF {$a'_m(k) \neq a_m(k)$}
  \STATE Share decision $a'_m(k)$ with other user requests;
  \ENDIF
  \ENDFOR
  \STATE Run a competitive mechanism for all user requests, and update strategy $a'(k)$ with the winning decision $a'_m(k)$;
  \IF {$\left | \Phi(a'(k))- \Phi(a(k))\right |< \epsilon $}
  \STATE $a^{*}=a'(k)$ and break;
  \ELSE
  \STATE $k \leftarrow k+1$;
  \ENDIF
  \ENDWHILE
  \RETURN $a^{*}$;
  \end{algorithmic}
\end{algorithm}

As the VNF placement problem is NP-hard \cite{7469866}, we implement a decentralized resource allocation algorithm by a potential game (PGRA) to find an approximate solution. The proposed PGRA algorithm can perform on the satellite network to improve the real-time decision-making capacity, where multiple satellites share the network resource states and the VNF placement strategies by interacting with each other, but satellites does not need to exchange the information with IoT users during the running time of the proposed PGRA algorithm. The procedure of the proposed PGRA algorithm is shown in Algorithm \ref{Resource Allocation Based on a Potential Game}. The maximum number of iterations is $K_{max}$. At the beginning, the initial iteration time is $k=0$, for $\forall u_m \in U$, we initialize the strategy as $a_m(k)=\varnothing$ and the user payoff as $\varphi_{m}(a_m(k),a_{-m}(k))$ accordingly. In iteration time $k$, each user request $u_m \in U$ needs to find an optimal strategy $a'_m(k)$ with maximum user payoff $\varphi_{m}(a'_m(k),a_{-m}(k))$ while the strategies of other user requests remain unchanged. For obtaining the optimal strategy of user request $u_m$, we search the $d$ shortest paths between source $s_m$ and destination $d_m$ successively and the local optimal strategy $a'_m(k,p)$ for each path $p$ is calculated by the Viterbi algorithm \cite{8933111}, which will be discussed later. Thus we can acquire an optimal strategy of user request $u_m$ in the current iteration by:
\begin{equation}\label{equation24}
a'_m(k) = \arg\max\limits_{a'_m(k,p)} \varphi (a'_m(k,p),a_{-m}(k)).
\end{equation}
If $a'_m(k) \neq a_m(k)$, the strategy $a'_m(k)$ will be shared by a message synchronization mechanism to other user requests for competing available resources of the satellite network. In each iteration, only a user request, which makes the overall network payoff maximum, can win the opportunity for updating its decision-making strategy and other user requests need to keep their old decision-making strategies. Note that the strategy decision-making processes for all user requests are performed simultaneously in parallel. The iteration process will terminate when no user request has an incentive to deviate its strategy unilaterally or the number of iterations exceeds the maximum iteration threshold. The final strategy profile $a^{*}=\left[a_{1}^{*},a_{2}^{*},\cdots, a_{M}^{*}\right]$ for all user requests is a Nash equilibrium of the VNF placement game and represents an approximate solution of VNF placement.
\begin{table*}[tbp]
  \centering
  \caption{Simulation Parameters.}
  \label{Simulation Parameters}
    \resizebox{\textwidth}{!}{
    \begin{tabular}{cccccccc}
    \hline
    \multicolumn{8}{c}{\bfseries Satellite Network} \\
    \hline
    Name  & \multicolumn{3}{c}{Total Number of Satellites} & \multicolumn{2}{c}{Number of Planes} & \multicolumn{2}{c}{Number of Satellites per Plane} \\
    Value & \multicolumn{3}{c}{6,9,12,15} & \multicolumn{2}{c}{3} & \multicolumn{2}{c}{2,3,4,5} \\
    \hline
    \multicolumn{8}{c}{\bfseries Inter-Satellite Links} \\
    \hline
    Name  & \multicolumn{3}{c}{Distance for a Plane} & \multicolumn{2}{c}{Distance for different Planes} & \multicolumn{2}{c}{Bandwidth} \\
    Value & \multicolumn{3}{c}{600 km} & \multicolumn{2}{c}{400 km} & \multicolumn{2}{c}{100 Mbps} \\
    \hline
    \multicolumn{8}{c}{\bfseries Edge Servers on Satellites} \\
    \hline
    Name  & vCPUs & Memory & Idle Power & Maximum Active Power & Setup Power & Maximum Idle Time & Minimum Off Time\\
    Value & 112   & 192 GB & 49.9 W & 415 W & 415 W & 3 slots & 1 slot \\
    \hline
    \multicolumn{8}{c}{\bfseries User Requests} \\
    \hline
    Name  & VNFs  & vCPUs & Memory & Execution Time for VNFs & Bandwidth & \multicolumn{2}{c}{Running Time} \\
    Lower & 5     & 4     & 4 GB  & 10 ms & 10 Mbps & \multicolumn{2}{c}{1 slot} \\
    Upper & 10    & 8     & 16 GB & 30 ms & 30 Mbps & \multicolumn{2}{c}{4 slots} \\
    \hline
    \end{tabular}%
    }
\end{table*}%
\begin{figure}[tbp]
  \centering
  \includegraphics[width = \columnwidth]{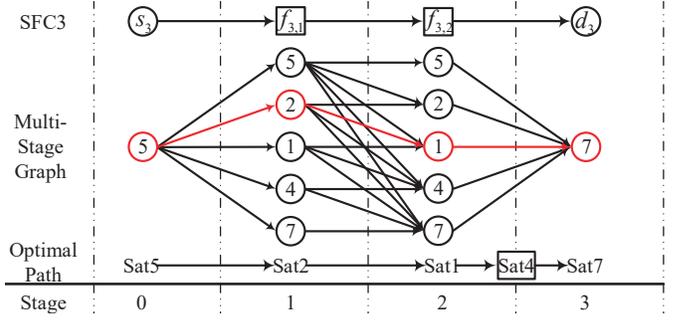}
  \caption{Example of placing VNFs for a user request by the Viterbi algorithm.}
  \label{Example of placing VNFs for a user request by the Viterbi algorithm}
\end{figure}

In this paper, we search the $d$ shortest source-to-destination paths for user request $\forall u_m \in U$ and use the Viterbi algorithm to place the VNFs for each path. the Viterbi algorithm can be viewed as a multi-stage graph for the states and their relationships and its aim is to find the most likely sequence of the states. The number of stages is equal to the length of an observed event sequence. Each node in a stage represents a possible state with a fixed cost for the current observed event and each stage in a graph consists of all possible states for the current observed event. The weighted edge between two states from adjacent stages represents a transmission cost. We can compute the cumulative cost of each state by the Viterbi algorithm in an increasing stage order, where the cumulative cost is composed of fixed costs and transmission costs. In the final stage, there are multiple possible states and each possible state corresponds to a path with a cumulative cost. We select a state path with minimum cumulative cost as the most likely path and obtain the best sequence by tracing the path. As the number of observed events and states increases the search space becomes large. Therefore, we can cut the search tree width to reduce the computational complexity.
\begin{algorithm}[tbp]
  \caption{Viterbi Algorithm.}
  \label{Viterbi Algorithm}
  \hspace*{0.02in} {\bf Input:} $u_m,p$;\\
  \hspace*{0.02in} {\bf Output:} $a'_m(k,p)$;
  \begin{algorithmic}[1]
  \STATE \textbf{Initialize:} $a'_m(k,p)=\varnothing,A_{layer}=\varnothing$;
  \STATE Obtain topological sort sequence $\Phi_{m}$ for the VNFs;
  \FOR {each $f_{m,i} \in \Phi_{m}$}
  \STATE ${A}'_{layer}=\varnothing$;
  \IF {$f_{m,i}=s_m$}
  \STATE $A_{layer}\leftarrow$ configure information about $s_m$;
  \STATE continue;
  \ENDIF
  \FOR {$\forall a_{m,i}^{pre}(k) \in A_{layer}$}
  \STATE Update network resource conditions under $a_{m,i}^{pre}(k)$;
  \STATE Obtain the set $\Omega_m$ of available satellites for path $p$;
  \FOR {each $u \in \Omega_m$}
  \STATE Deploy $f_{m,i}$ to satellite $u$ by strategy $a_{m,i}^{curr}(k)$;
  \IF {the constraints in \eqref{equation11}-\eqref{equation18} are satisfied}
  \STATE ${A}'_{layer}\leftarrow \left [a_{m,i}^{pre}(k), a_{m,i}^{curr}(k)\right ]$;
  \ENDIF
  \ENDFOR
  \ENDFOR
  \STATE List ${A}'_{layer}$ by user payoffs in descending order.
  \STATE ${A_{layer}\leftarrow A}'_{layer}[:B]$;
  \ENDFOR
  \STATE Obtain strategy $a'_m(k,p)$ with maximum user payoff;
  \RETURN $a'_m(k,p)$;
  \end{algorithmic}
\end{algorithm}

For finding an approximate solution of the VNF placement by the Viterbi algorithm, we construct the VNF placement states and their relationships as a multi-stage graph. For a user request, the SFC is considered as an observed event sequence, each VNF indicates an observed event and the number of the VNFs represents the number of stages. The strategy of deploying a VNF indicates a possible state in each stage and there has the VNF deployment cost accordingly. The path between two adjacent VNFs indicates the edge between two states from two adjacent stages and the path cost is equal to the edge cost. Thus, we can perform the Viterbi algorithm to obtain an approximate strategy of VNF placement for a user request. Fig.~\ref{Example of placing VNFs for a user request by the Viterbi algorithm} illustrates an example of placing VNFs for a user request by the Viterbi algorithm, where the approximate VNF placement strategy is indicated with a red line.

The Viterbi algorithm for the VNF placement is shown in Algorithm \ref{Viterbi Algorithm}. The input parameters are user request $u_m$ and path $p$. The output parameter is an approximate strategy $a'_m(k,p)$ for path $p$. Initially, we denote a state set of the first stage by $A_{layer}=\varnothing$ and set $a'_m(k,p)=\varnothing$. Then we can obtain the ordered VNF sequence $\Gamma_m$ by a topological sorting method. For each stage of VNF $f_{m,i}\in \Gamma_m$, we search all possible VNF placement states and calculate their cumulative costs. respectively. If $f_{m,i}=s_m$, we can directly update the configure information concerning $s_m$ to $A_{layer}$. If $f_{m,i} \neq s_m$, for $\forall a_{m,i}^{pre}(k) \in A_{layer}$, we first update network resource conditions by $a_{m,i}^{pre}(k)$ and obtain a set $\Omega_{m,i}$ of current available satellites for path $p$. Within the network resource and service requirement constraints, we deploy VNF $f_{m,i}$ to satellite $u \in \Omega_{m,i}$ by strategy $a_{m,i}^{curr}(k)$. When the constraints in equations \eqref{equation11}-\eqref{equation18} are satisfied, $a_{m,i}^{curr}(k)$ will be put into a new strategy set $A'_{layer}$, which is initialized as $A'_{layer}=\varnothing$ at the beginning of each stage. When each stage is over, we will sort $A'_{layer}$ by user payoffs in descending order and use the first $B$ paths from $A'_{layer}$ to update $A_{layer}$. If all VNFs are deployed we can obtain an approximate strategy $a'_m(k,p)$ with maximum user payoff.

For Algorithm \ref{Resource Allocation Based on a Potential Game}, the computation complexity can be indicated as $O(K_{max}Md)$, where $K_{max}$ is the maximum number of iterations, $M$ is the number of players for the current time slot, and $d$ represents the size of a candidate path set. When we perform the Viterbi algorithm to deploy the VNFs to satellite nodes, the search tree width is $B$, the number of satellite nodes for an available path is less than $N$, and the maximum number of VNFs for a user request is $F$. The computation complexity of Algorithm \ref{Viterbi Algorithm} can be described as $O(FBN)$.

\section{Performance Evaluation}\label{Performance Evaluation}
In this section, we make the experiments to evaluate the performance of the proposed PGRA algorithm for addressing the VNF placement problem in satellite edge computing. We setup the system parameters for performance evaluation. In order to investigate the effects of system parameters on the performance of the proposed PGRA algorithm, we design the experiments by the Taguchi method with two factors, which are the shortest paths $d$ between the source and the destination for a user request and the width $B$ of the Viterbi search tree. Furthermore, we compare the proposed PGRA algorithm with two existing centralized baseline algorithms of Viterbi \cite{7469866} and Greedy \cite{7332796} in terms of network payoff and percentage of allocated users. Finally, we discuss the performance of the proposed PGRA algorithm in on-line strategy.

\subsection{Simulation Setup}
\begin{figure*}[tbp]
  \centering
  \subfigure[Main effects for $M=10$]{\includegraphics[width=0.3\textwidth]{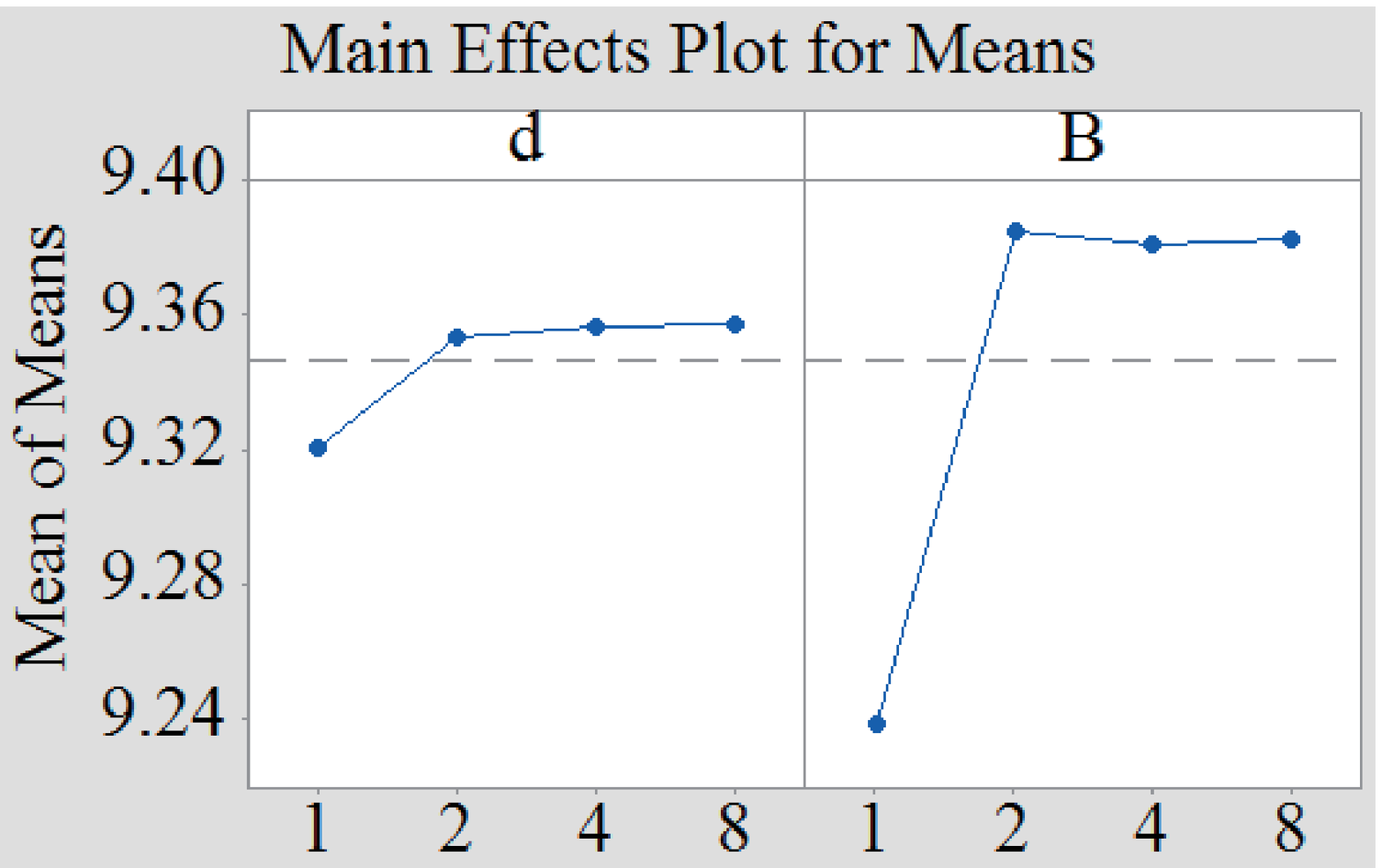}
  \label{Main effects for 30 user requests}}
  \subfigure[Main effects for $M=20$]{\includegraphics[width=0.3\textwidth]{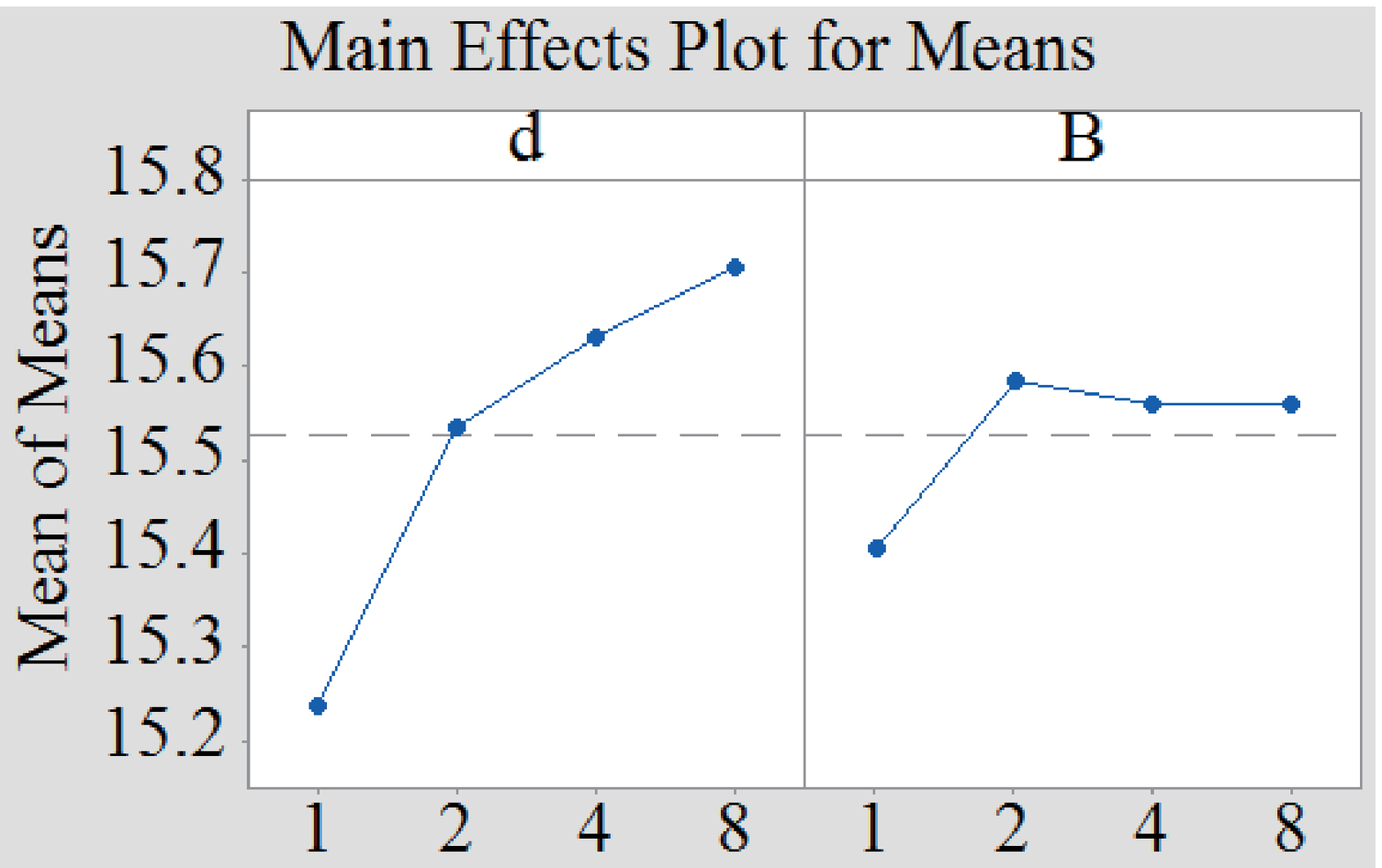}
  \label{Main effects for 30 user requests}}
  \subfigure[Main effects for $M=30$]{\includegraphics[width=0.3\textwidth]{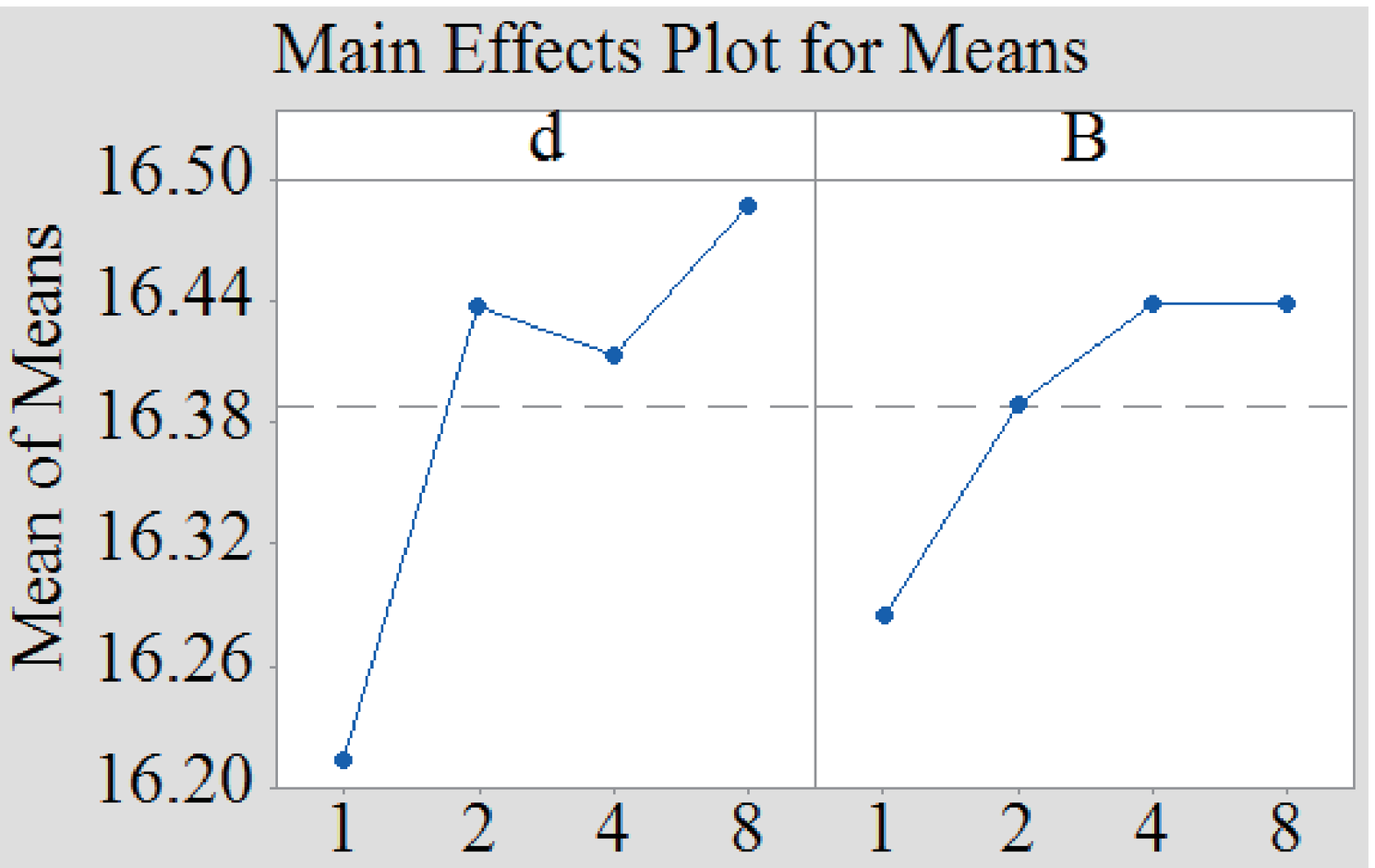}
  \label{Main effects for 30 user requests}}
  \caption{Main effects of two factors for different number $M=10,20,$ and $30$ of user requests.}
  \label{Main effects of two factors for different number of user requests}
\end{figure*}
In the simulation experiments, we set the number of LEO satellite nodes by 6, 9, 12, and 15, respectively, where there are 3 orbital planes and each plane consists of 2, 3, 4, and 5 satellite nodes accordingly. There is an edge server on each satellite node. For each edge server, the resource capacities are configured as 112 vCPUs and 192 GB Memory, we consider the idle power as 49.9 W and the maximum active power as 415 W \cite{2019SPECpower}. We assume that the maximum idle time interval is 3 time slots and the minimum off time interval is 1 time slot. In addition, we assume that the inter-satellite link distance for the same orbital plane is 600 km and for the different orbital planes is 400 km, respectively. The bandwidth capacity for each inter-satellite link is 100 Mbps.
\begin{table}[tbp]
  \centering
  \caption{Parameters for Taguchi Method.}
  \label{Parameters for Taguchi Method}%
  \resizebox{\columnwidth}{!}{
    \begin{tabular}{p{0.08\textwidth}<{\centering}p{0.07\textwidth}<{\centering}p{0.07\textwidth}<{\centering}p{0.07\textwidth}<{\centering}p{0.07\textwidth}<{\centering}}
    \hline
    \multirow{2}{*}{Factor} & \multicolumn{4}{c}{Level} \\
\cline{2-5}          & 1     & 2     & 3     & 4 \\
    \hline
    d     & 1     & 2     & 4     & 8 \\

    B     & 1     & 2     & 4     & 8 \\
    \hline
    \end{tabular}
    }
\end{table}%
\begin{table}[tbp]
  \centering
  \caption{Orthogonal Table $L_{16}(4^2)$ and Network Payoff Results.}
  \label{Orthogonal Table and Network Payoff Results}%
    \resizebox{\columnwidth}{!}{
    \begin{tabular}{p{0.05\textwidth}<{\centering}p{0.04\textwidth}<{\centering}p{0.04\textwidth}<{\centering}p{0.07\textwidth}<{\centering}p{0.07\textwidth}<{\centering}p{0.07\textwidth}<{\centering}}
    \hline
    \multirow{2}{*}{Number} & \multicolumn{2}{c}{Factor} & \multicolumn{3}{c}{Network Payoff} \\
\cline{2-6}          & d     & B     & \multicolumn{1}{c}{M=10} & \multicolumn{1}{c}{M=20} & \multicolumn{1}{c}{M=30} \\
    \hline
    0     & 1     & 1     & 9.1585 & 15.1844 & 16.1857 \\
    1     & 1     & 2     & 9.3734 & 15.1888 & 16.2896 \\
    2     & 1     & 4     & 9.3748 & 15.2886 & 16.1899 \\
    3     & 1     & 8     & 9.3751 & 15.2891 & 16.1899 \\
    4     & 2     & 1     & 9.2622 & 15.3828 & 16.3832 \\
    5     & 2     & 2     & 9.3822 & 15.5843 & 16.3890 \\
    6     & 2     & 4     & 9.3832 & 15.5844 & 16.4879 \\
    7     & 2     & 8     & 9.3846 & 15.5844 & 16.4879 \\
    8     & 4     & 1     & 9.2680 & 15.4805 & 16.2845 \\
    9     & 4     & 2     & 9.3916 & 15.6820 & 16.3895 \\
    10    & 4     & 4     & 9.3825 & 15.6821 & 16.4884 \\
    11    & 4     & 8     & 9.3838 & 15.6814 & 16.4884 \\
    12    & 8     & 1     & 9.2680 & 15.5785 & 16.2846 \\
    13    & 8     & 2     & 9.3919 & 15.8785 & 16.4872 \\
    14    & 8     & 4     & 9.3828 & 15.6819 & 16.5862 \\
    15    & 8     & 8     & 9.3840 & 15.6812 & 16.5862 \\
    \hline
    \end{tabular}%
    }
\end{table}%

In order not to lose generality, we randomly generate the SFC and resource requirements for each user request. The number of VNFs is from 5 to 10. Each VNF, except source and destination, has a predecessor and a successor, and requires $\left[4,8\right]$ vCPUs and $\left[4 GB,16 GB \right]$ Memory, respectively. The execution time for each VNF is $\left[10,30\right]$ ms. The bandwidth demand between two adjacent VNFs for routing traffic flows is $\left[10,30\right]$ Mbps. The source and the destination are randomly generated from the set of satellite nodes and can be known in advance. In addition, we define the weighted values in equation \eqref{equation9} as $\alpha_{1}=\alpha_{2}=\alpha_{3}=\frac{1}{3}$. Table \ref{Simulation Parameters} summarizes the main simulation parameters in our evaluation. We use a commodity server to be the simulation platform, where the configuration information is i7-4790K CPU, 16 GB Memory, and Windows 10. The programming language is PYTHON.

\subsection{System Parameters Evaluation}
For the proposed PGRA algorithm, the performance results of addressing the VNF placement problem can be influenced by two important parameters, which are the shortest paths $d$ between the source and the destination for a user request and the width $B$ of the Viterbi search tree, respectively. Consequently, in a satellite network with 6 satellite nodes, we use the Taguchi method of design-of-experiment (DOE) \cite{2013An} to investigate the simulation results under different values of $d$ and $B$ \cite{2015A}. The two factors are denoted by $d$ and $B$, respectively, where each factor includes 4 levels, e.g., 1, 2, 4, and 8. The parameters for the Taguchi method \cite{2013An} are shown in Table \ref{Parameters for Taguchi Method}. The orthogonal table $L_{16}(4^2)$ consists of 16 instances and we run each instance, for $M=10$, 20, and 30, 10 times to obtain the average network payoffs, respectively. Table \ref{Orthogonal Table and Network Payoff Results} describes the orthogonal table and network payoff results in our simulation. The main effects of the two factors for different user requests are illustrated in Fig. \ref{Main effects of two factors for different number of user requests}. We can find from Fig. \ref{Main effects of two factors for different number of user requests} that the network payoff results for $M=10$, 20, and 30 are better as parameters $d$ and $B$ increase. Large $d$ can improve the exploration ability of the proposed PGRA algorithm by searching for a larger solution space. For the Viterbi algorithm, large $B$ can keep more possible states of the current stage to the next search stage and also promote the network performance. However, large $d$ and $B$ can lead to a high computational complexity of the proposed PGRA algorithm. When $d$ and $B$ are small, the performance of the proposed PGAR algorithm will degrade. Therefore, their values should be considered in a tradeoff way. According to the simulation results of the Taguchi method, we can find that the ideal values of $d$ and $B$ are 8 and 4, respectively.
\begin{figure*}[tbp]
  \centering
  \subfigure[Average energy cost]{\includegraphics[width=0.3\textwidth]{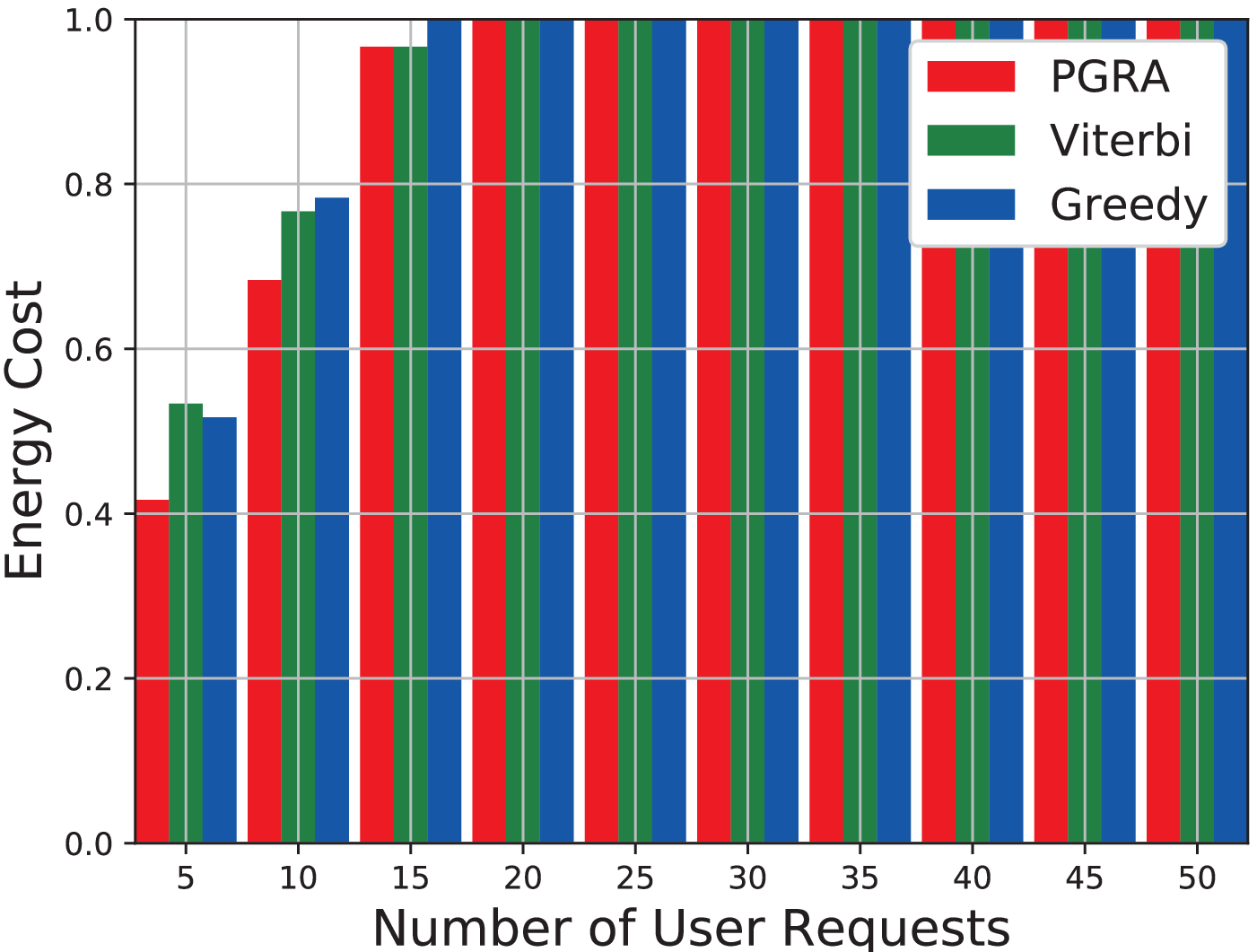}
  \label{Average energy cost}}
  \subfigure[Average bandwidth cost]{\includegraphics[width=0.3\textwidth]{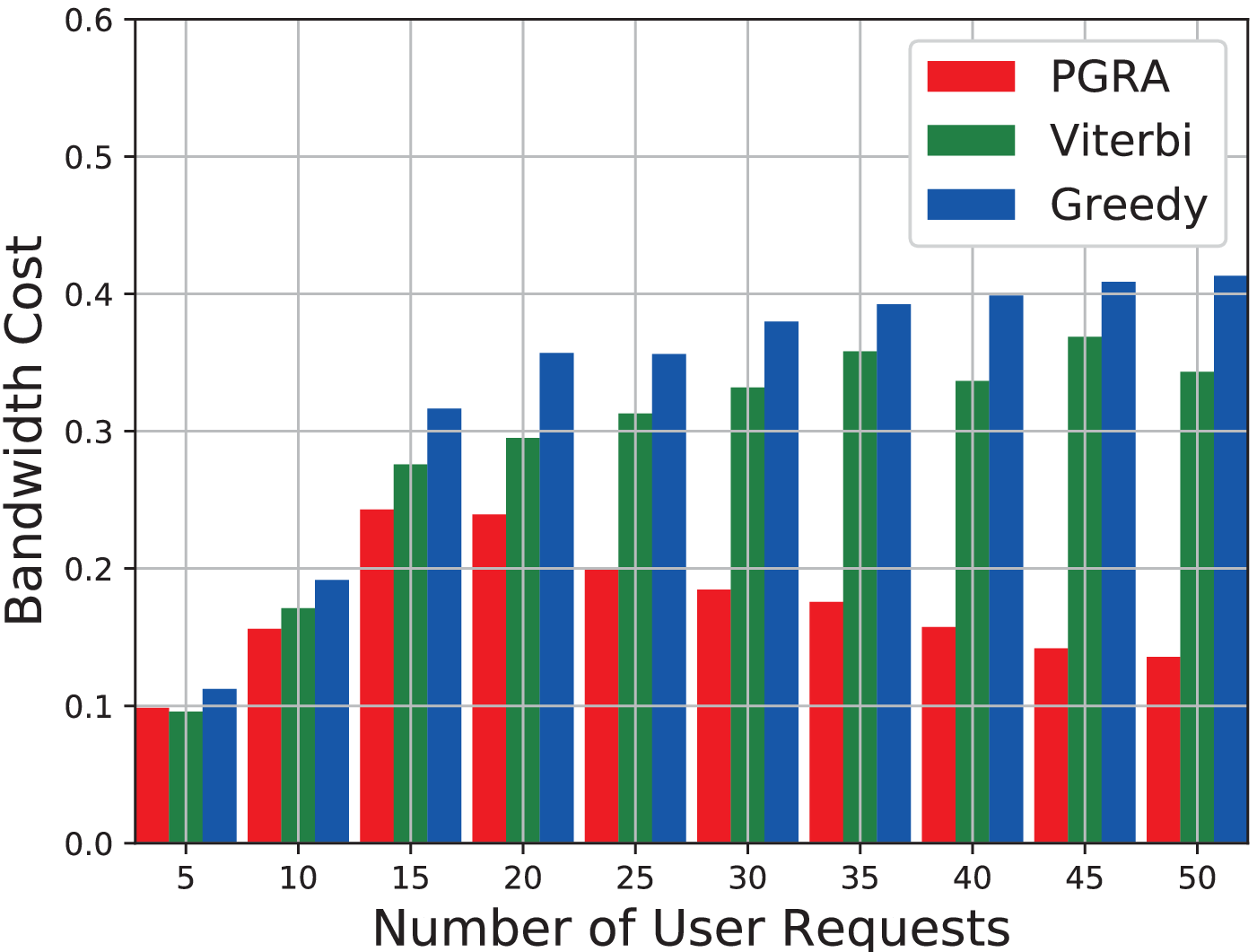}
  \label{Average bandwidth cost}}
  \subfigure[Average service delay cost]{\includegraphics[width=0.3\textwidth]{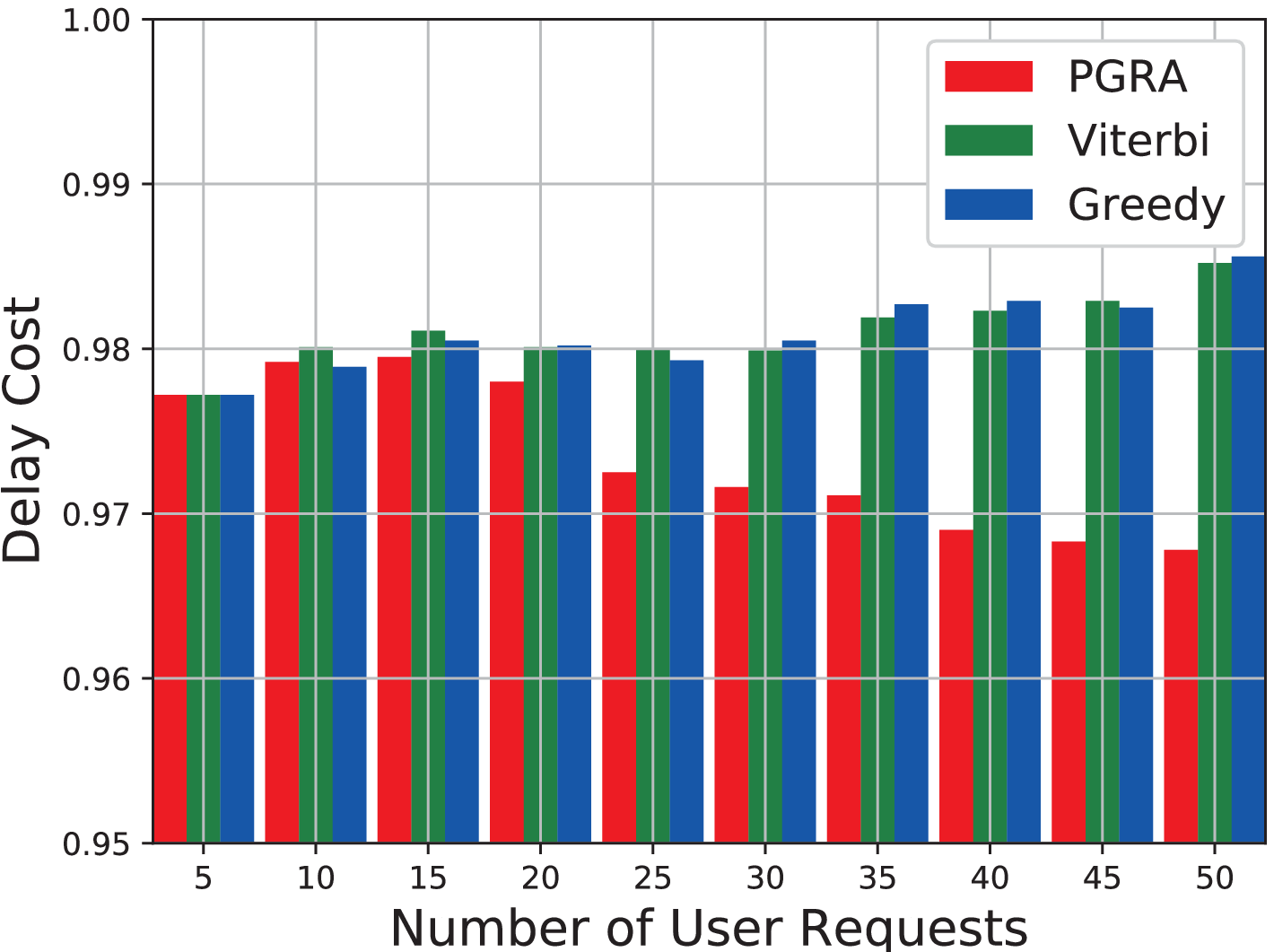}
  \label{Average service delay cost}}
  \subfigure[Average network payoff]{\includegraphics[width=0.3\textwidth]{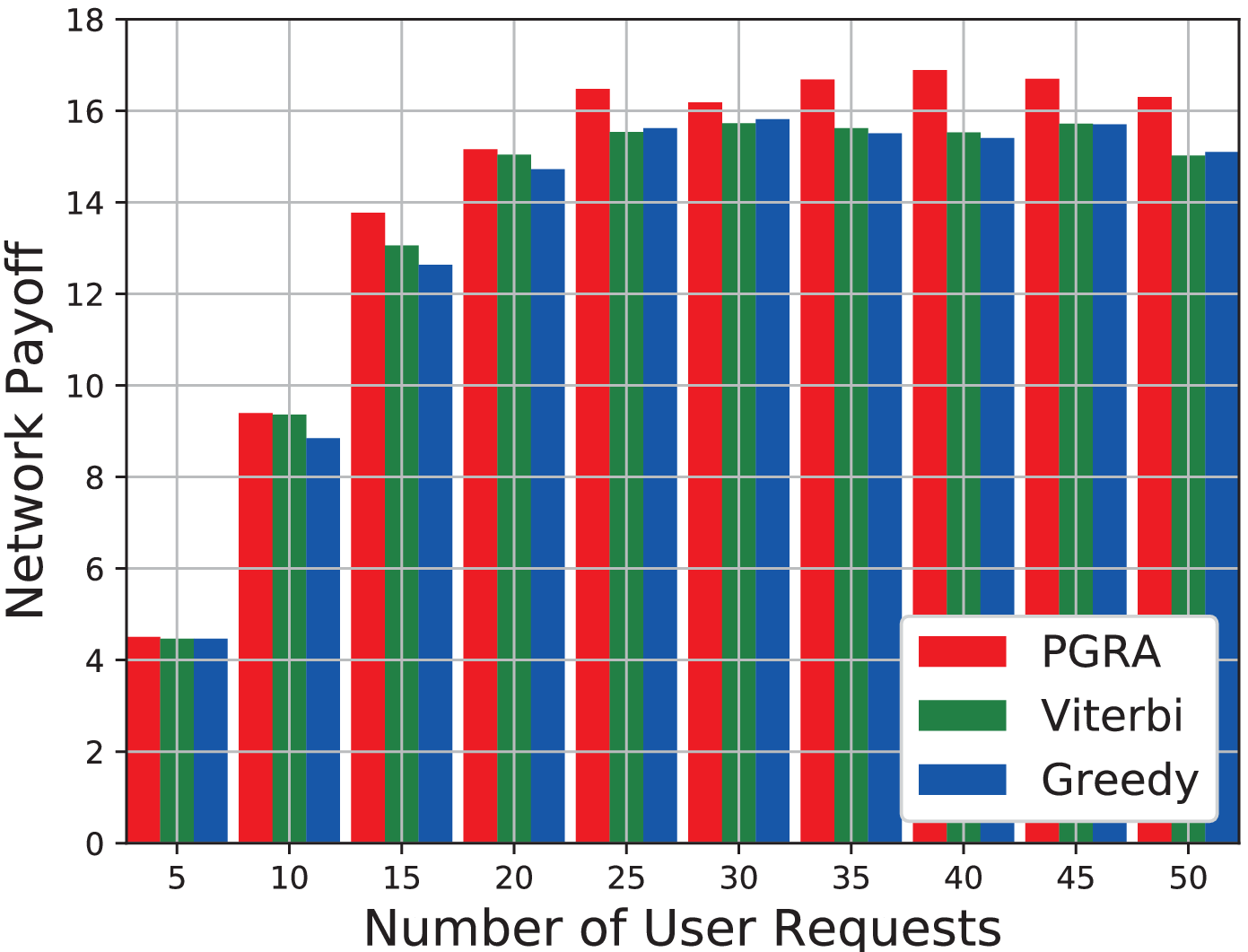}
  \label{Average network payoff}}
  \subfigure[Average percentage of allocated users]{\includegraphics[width=0.3\textwidth]{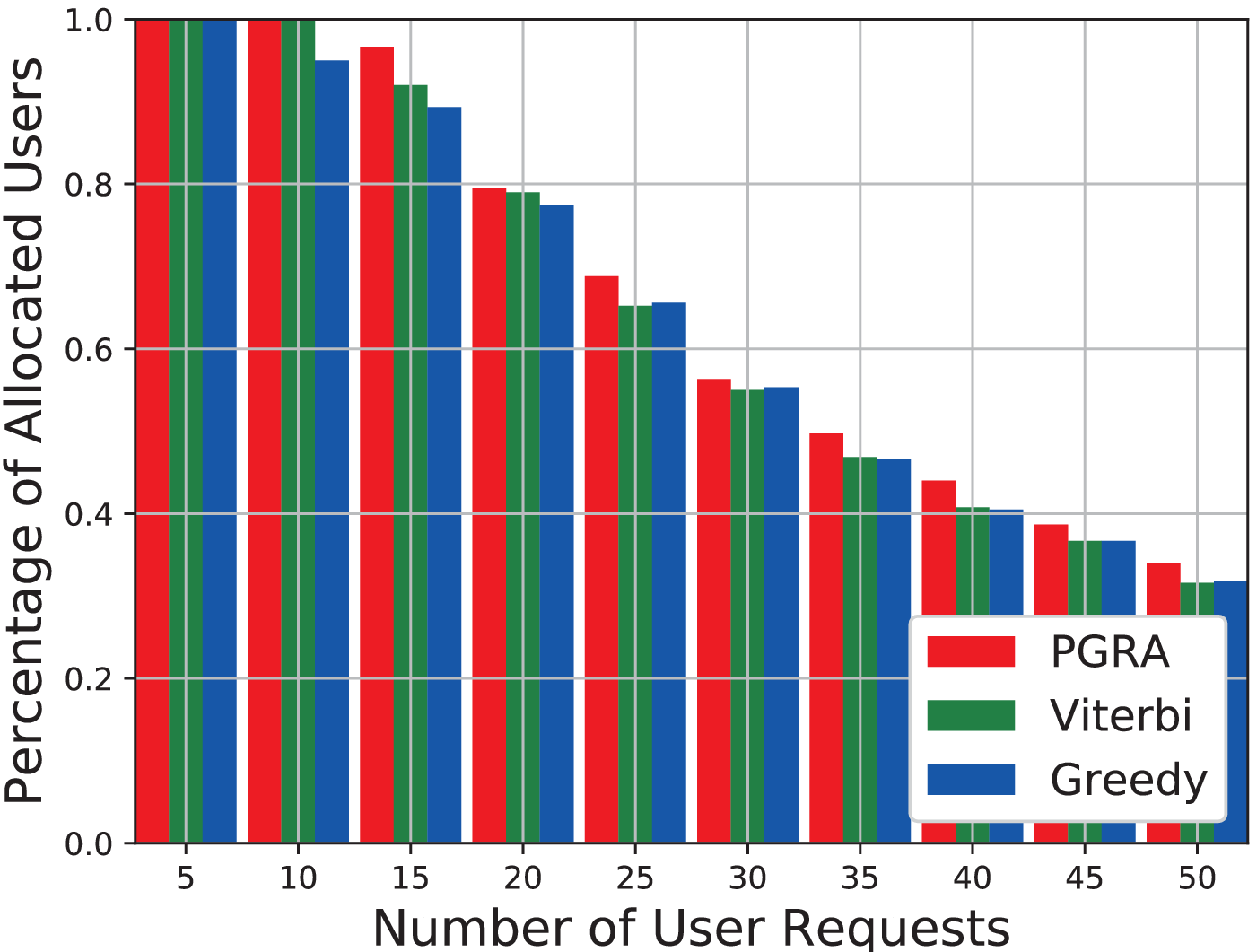}
  \label{Average percentage of allocated users}}
  \caption{Performance comparison with PGRA, Viterbi, and Greedy in a satellite network with 6 satellite nodes.}
  \label{Performance comparison with PGRA, Viterbi, and Greedy in a satellite network with 6 satellite nodes}
\end{figure*}

\subsection{Performance Comparison with Baseline Algorithms}
To further discuss the effectiveness of the proposed PGRA algorithm in terms of energy consumption, bandwidth, and service delay, we compare the proposed PGRA algorithm with two existing baseline algorithms of Viterbi \cite{7469866} and Greedy \cite{7332796}. Based on the parameter analysis of the Taguchi method, the values of $d$ and $B$ are set as 8 and 4, respectively. A satellite network with 6 satellite nodes is used to run the experiment $L=\left \{L_1,L_2,\cdots,L_{10}\right \}$, which is composed of 10 group experiments and the corresponding number of user requests is denoted by $\left \{5,10,\cdots,5*i,\cdots,50\right \}$. Each experiment is run 10 times and the average results are obtained.

The experiment results for the performance comparison with three optimization algorithms are shown in Fig.~\ref{Performance comparison with PGRA, Viterbi, and Greedy in a satellite network with 6 satellite nodes}. Fig.~\ref{Average energy cost} illustrates the average energy costs for different number of user requests. We can find from Fig.~\ref{Average energy cost} that the energy costs obtained by the proposed PGRA algorithm are better than that of Viterbi and Greedy when there is a small number of user requests, e.g., $M=5$, 10, and 15. For instance $L_2$ with 10 user requests, the average energy costs for PGRA, Viterbi, and Greedy are 0.6833, 0.7666, and 0.7833, respectively. The average energy cost obtained by the proposed PGRA algorithm reduces by $10.87\%$ for Viterbi and $12.77\%$ for Greedy. However, as the number of user requests increases, the energy costs are increasing until they reach the maximum values. Then new user requests can not be deployed to satellite nodes due to the limitation of network resource capacities.

The average bandwidth costs for deploying different user requests to satellite nodes are described in Fig.~\ref{Average bandwidth cost}. We can observe from Fig.~\ref{Average bandwidth cost} that the proposed PGRA algorithm performs better than the two baseline algorithms of Viterbi and Greedy. For the small number of user requests, such as $M=5$, 10, and 15, the performance of the proposed PGRA algorithm is slightly better than that of Viterbi and Greedy. For an example of $M=10$, the average bandwidth costs are 0.1560, 0.1711, and 0.1915 for PGRA, Viterbi, and Greedy, respectively. The performance improvement of the proposed PGRA algorithm is $8.82\%$ for Viterbi and $18.53\%$ for Greedy. As the number of user requests increases, we can observe that the performance differences between the proposed PGRA algorithm and the two baseline algorithms are also increasing. In the case of $M=35$, the average bandwidth costs for PGRA, Viterbi, and Greedy are 0.1756, 0.3581, and 0.3923, respectively. We can observe that the performance of the proposed PGRA algorithm improves by $50.96\%$ for Viterbi and $55.23\%$ for Greedy. That is due to the fact that the players in a potential game want to make decisions by competing with each others for optimizing their objectives in a self-interested behavior. Thus, under the limitation of network resource capacities, the players, which have better strategies of placing the VNFs, can win the opportunities of updating their strategy information. In our experiments, when there are enough user requests to require available network resources, these resource requirements are greater than the network resource capacities and the energy costs produced by the user requests tend to constant values. Therefore, the used bandwidth resources can be considered as an important optimization objective. For the 10 group experiments, the average performance improvement of the proposed PGRA algorithm is $40.05\%$ for Viterbi and $47.93\%$ for Greedy. In Fig.~\ref{Average service delay cost}, Average service delay costs for different user requests are indicated. Similar to the average bandwidth costs in Fig.~\ref{Average bandwidth cost}, we can observe that the average service delay costs obtained by the proposed PGRA algorithm are better than that of the two baseline algorithms. On average, the service delay costs for PGRA, Viterbi, and Greedy are 0.9734, 0.9811, and 0.9810, respectively. The performance of the proposed PGRA algorithm improves by $0.78\%$ over both Viterbi and Greedy.
\begin{figure*}[tbp]
  \centering
  \subfigure[Satellite network with $N=9$]{\includegraphics[width=0.3\textwidth]{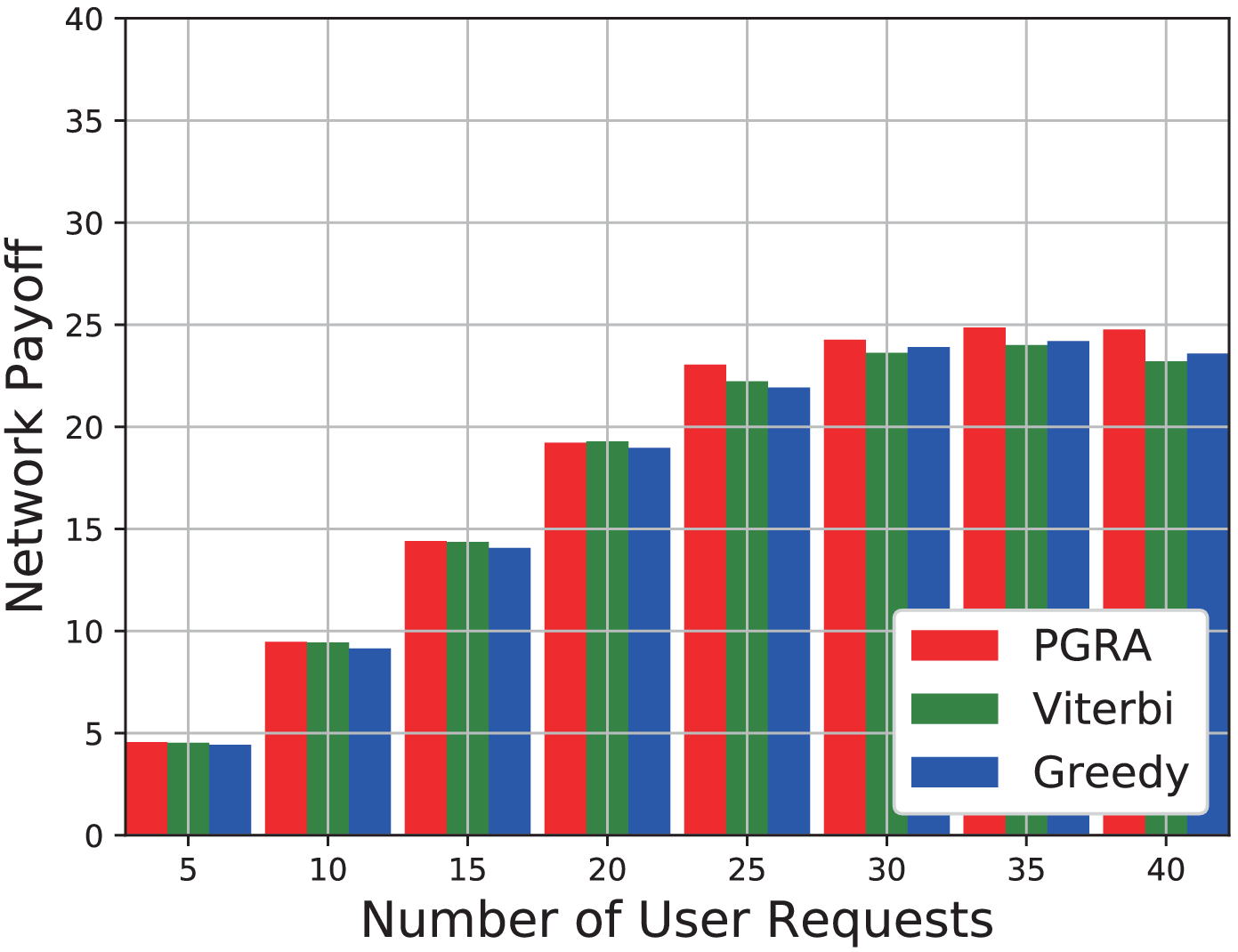}
  \label{Satellite network with N=9}}
  \subfigure[Satellite network with $N=12$]{\includegraphics[width=0.3\textwidth]{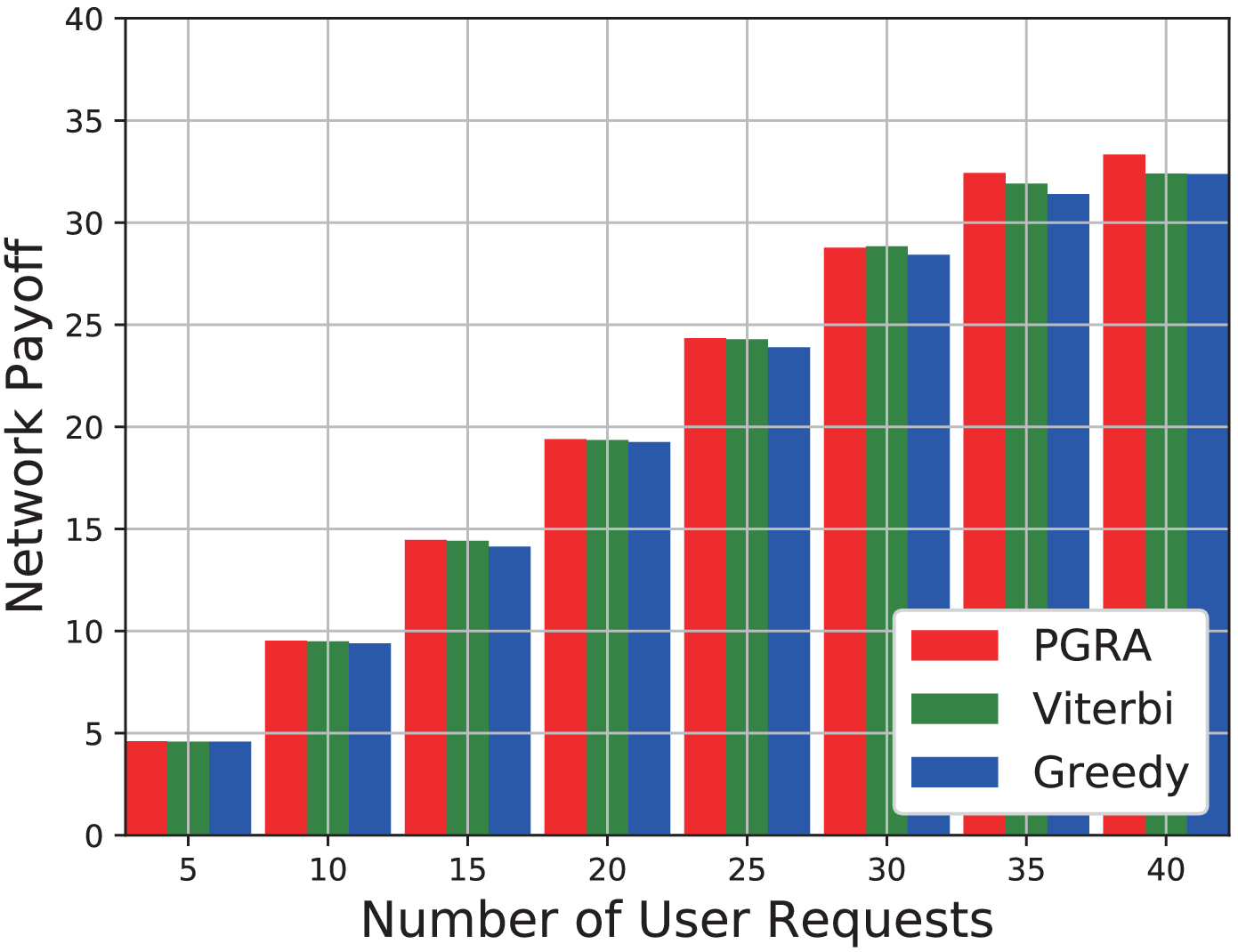}
  \label{Satellite network with N=12}}
  \subfigure[Satellite network with $N=15$]{\includegraphics[width=0.3\textwidth]{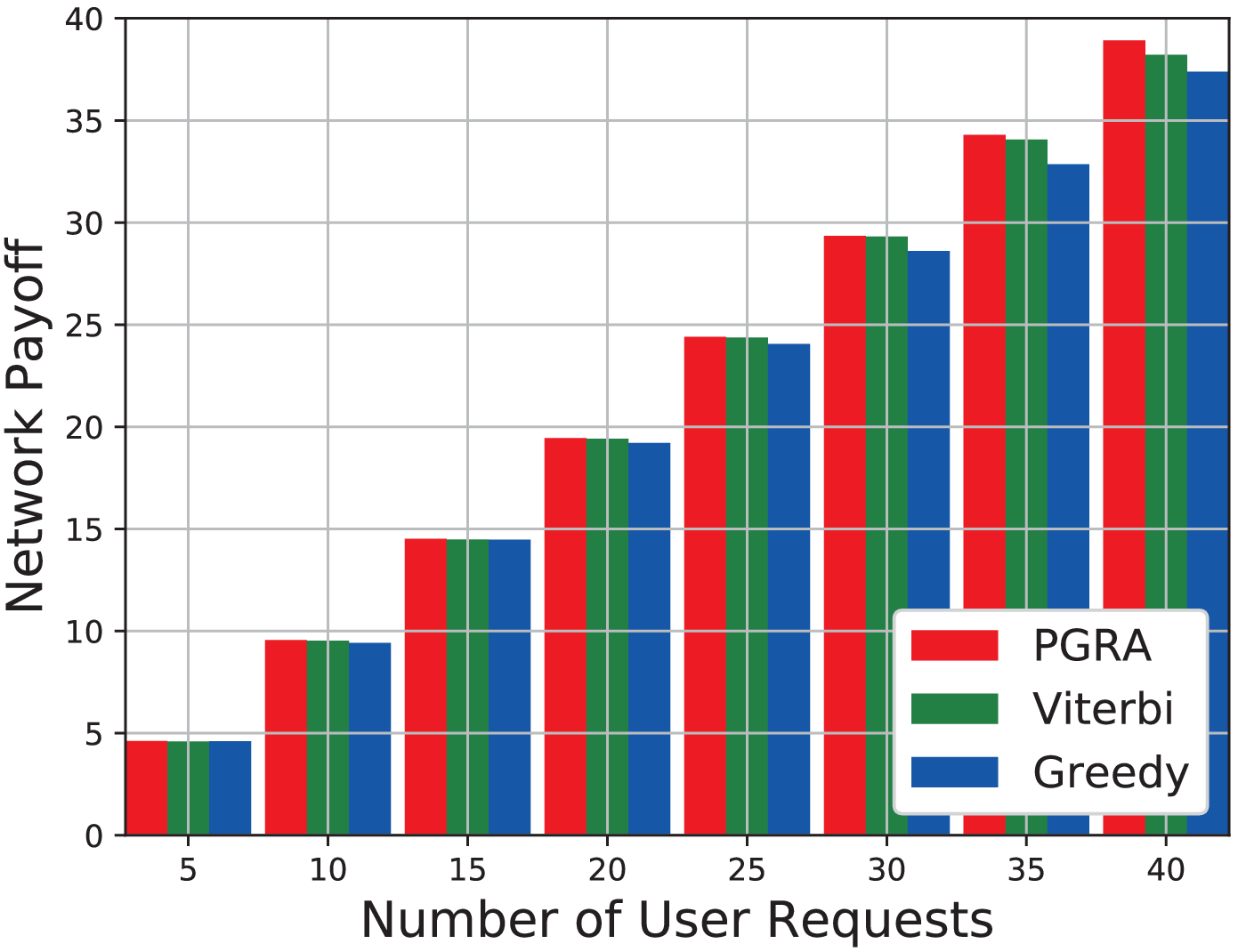}
  \label{Satellite network with N=15}}
  \caption{Network payoffs for PGRA, Viterbi, and Greedy in satellite networks with $N=9$, 12, and 15.}
  \label{Network payoffs for PGRA, Viterbi, and Greedy in satellite networks}
\end{figure*}

Fig.~\ref{Average network payoff} describes the average network payoffs for different number of user requests. It can be found from Fig.~\ref{Average network payoff} that the proposed PGRA algorithm outperforms the two baseline algorithms of Viterbi and Greedy in all the experiments. For an example of $M=15$, the network payoffs for PGRA, Viterbi, and Greedy are 13.7703, 13.0588, and 12.6343, respectively, and the result obtained by the proposed PGRA algorithm is over $5.44\%$ for Viterbi and $8.99\%$ for Greedy. Overall, the average performance improvement of the proposed PGRA algorithm is $5.16\%$ for Viterbi and $6.15\%$ for Greedy, respectively. In Fig.~\ref{Average percentage of allocated users}, we illustrate the average percentages of allocated user requests. We can find from Fig.~\ref{Average percentage of allocated users} that all user requests can be deployed to satellite nodes when the number of user requests is small, e.g., $M=10$. As $M$ increases, the percentage of allocated user requests begins to decrease due to the resource limitation of a satellite network. In these cases, the proposed PGRA algorithm also outperforms Viterbi and Greedy in terms of the percentage of allocated user requests. On average, the performance of the proposed PGRA algorithm is better $3.18\%$ than Viterbi and $4.60\%$ than Greedy. From Fig.~\ref{Performance comparison with PGRA, Viterbi, and Greedy in a satellite network with 6 satellite nodes}, we can demonstrate the effectiveness of the proposed PGRA algorithm compared with two baseline algorithms of Viterbi and Greedy, meanwhile, it is shown that the proposed PGRA algorithm outperforms Viterbi and Greedy for deploying user requests to satellite nodes.

In order to further evaluate the performance of the proposed PGRA algorithm in different scale satellite networks, three satellite networks, which consist of $N=9$, 12, and 15 satellite nodes, are used for conducting the experiments with user requests $\left \{5,10,15,20,25,30,35,40\right \}$, respectively. Each experiment is run 10 times and we obtain the average results. Fig.~\ref{Network payoffs for PGRA, Viterbi, and Greedy in satellite networks} shows the network payoffs for different user requests in three satellite networks with $N=9$, 12, and 15, respectively. In Fig.~\ref{Satellite network with N=9}, we describe the network payoffs for PGRA, Viterbi, and Greedy in the satellite network with $N=9$. Overall, we can find that the proposed PGRA algorithm performs better than Viterbi and Greedy. When the number of user requests is small, the performance of the proposed PGRA algorithm is close to that of Viterbi and over that of Greedy. As the number of user requests increases, the network payoff achieved by the proposed PGRA algorithm is more than that of Viterbi and Greedy. For example, when $M=10$, the network payoffs for PGRA, Viterbi, and Greedy are 9.4731, 9.4428, and 9.1489, respectively. The performance improvement of the proposed PGRA algorithm is $0.32\%$ for Viterbi and $3.54\%$ for Greedy. When $M=25$, the network payoffs for PGRA, Viterbi, and Greedy are 23.0444, 22.2306, and 21.9208, respectively. The proposed PGRA algorithm performs better $3.66\%$ than Viterbi and $5.12\%$ than Greedy. We can also find that the network payoffs begin to be stable when the number of user requests is greater than $25$. That is due to the fact that more user requests can not be deployed to satellite nodes as the resource limitation of a satellite network. On average, the performance of the proposed PGRA algorithm improves by $2.78\%$ for Viterbi and $3.11\%$ for Greedy. Fig.~\ref{Satellite network with N=12} and Fig.~\ref{Satellite network with N=15} illustrate the network payoffs for PGRA, Viterbi, and Greedy in satellite networks with $N=12$ and 15, respectively. Similar to the results in Fig.~\ref{Satellite network with N=9}, the proposed PGRA algorithm outperforms Viterbi and Greedy in satellite networks with $N=12$ and 15. For $N=12$, the network payoff obtained by the proposed PGRA algorithm increases by $0.96\%$ for Viterbi and $2.08\%$ for Greedy. For $N=15$, the network payoff obtained by the proposed PGRA algorithm increases by $0.65\%$ for Viterbi and $2.64\%$ for Greedy. In addition, as the increasing number in satellite nodes, a satellite network can provide more available network resources for user requests. Thus, more user requests can be placed to satellite nodes and that can lead to an increase in network payoff. For an instance of $M=35$, the network payoffs obtained by the proposed PGRA algorithm are 24.8568, 32.4339, and 34.2962 for satellite networks with $N=9$, 12, and 15, respectively. The network payoff obtained by the proposed PGRA algorithm for the satellite network with $N=15$ is over $37.97\%$ for the satellite network with $N=9$ and $5.74\%$ for the satellite network with $N=12$.
\begin{figure*}[tbp]
  \centering
  \subfigure[Satellite network with $N=9$]{\includegraphics[width=0.3\textwidth]{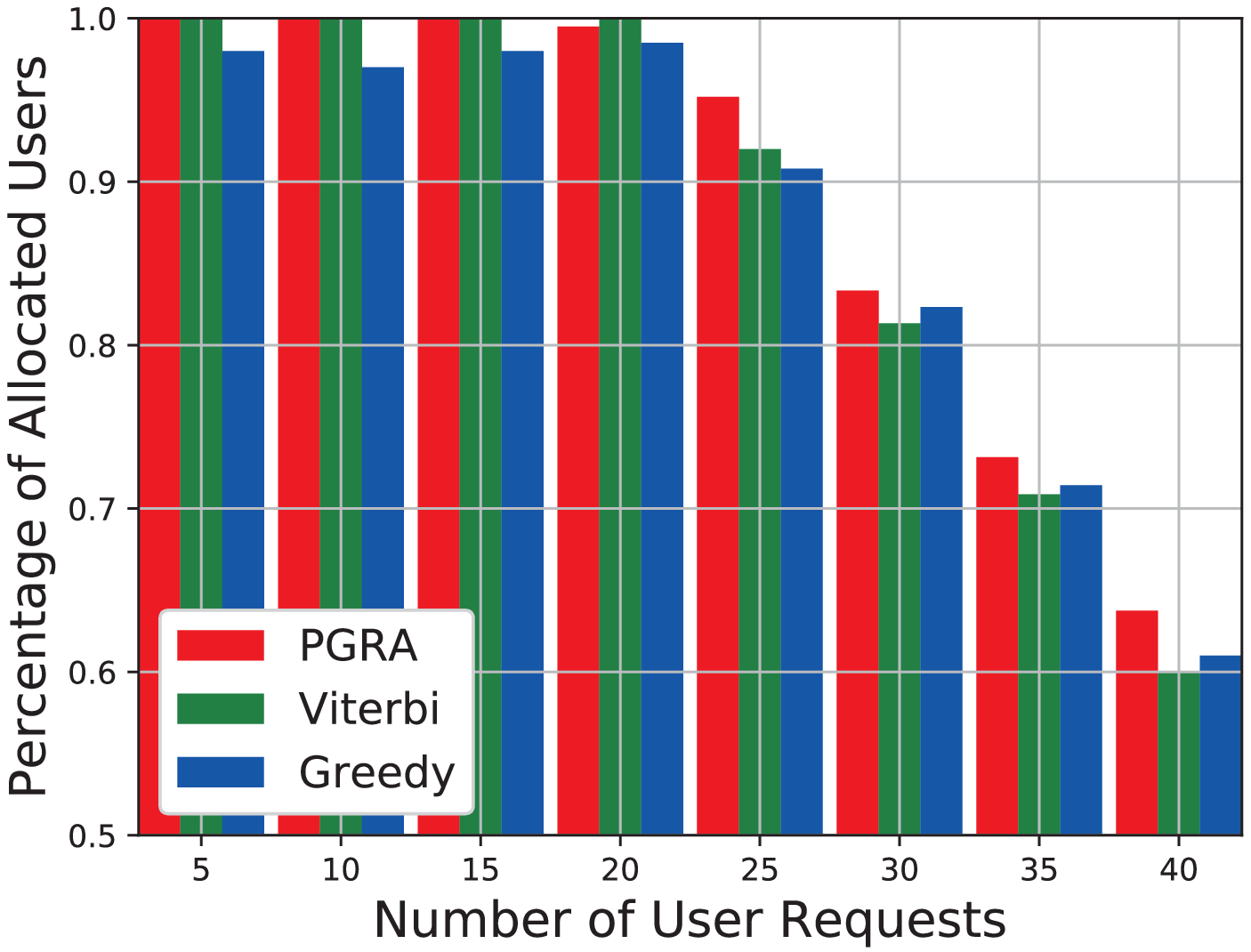}
  \label{Percentages of allocated user requests in a satellite network with N=9}}
  \subfigure[Satellite network with $N=12$]{\includegraphics[width=0.3\textwidth]{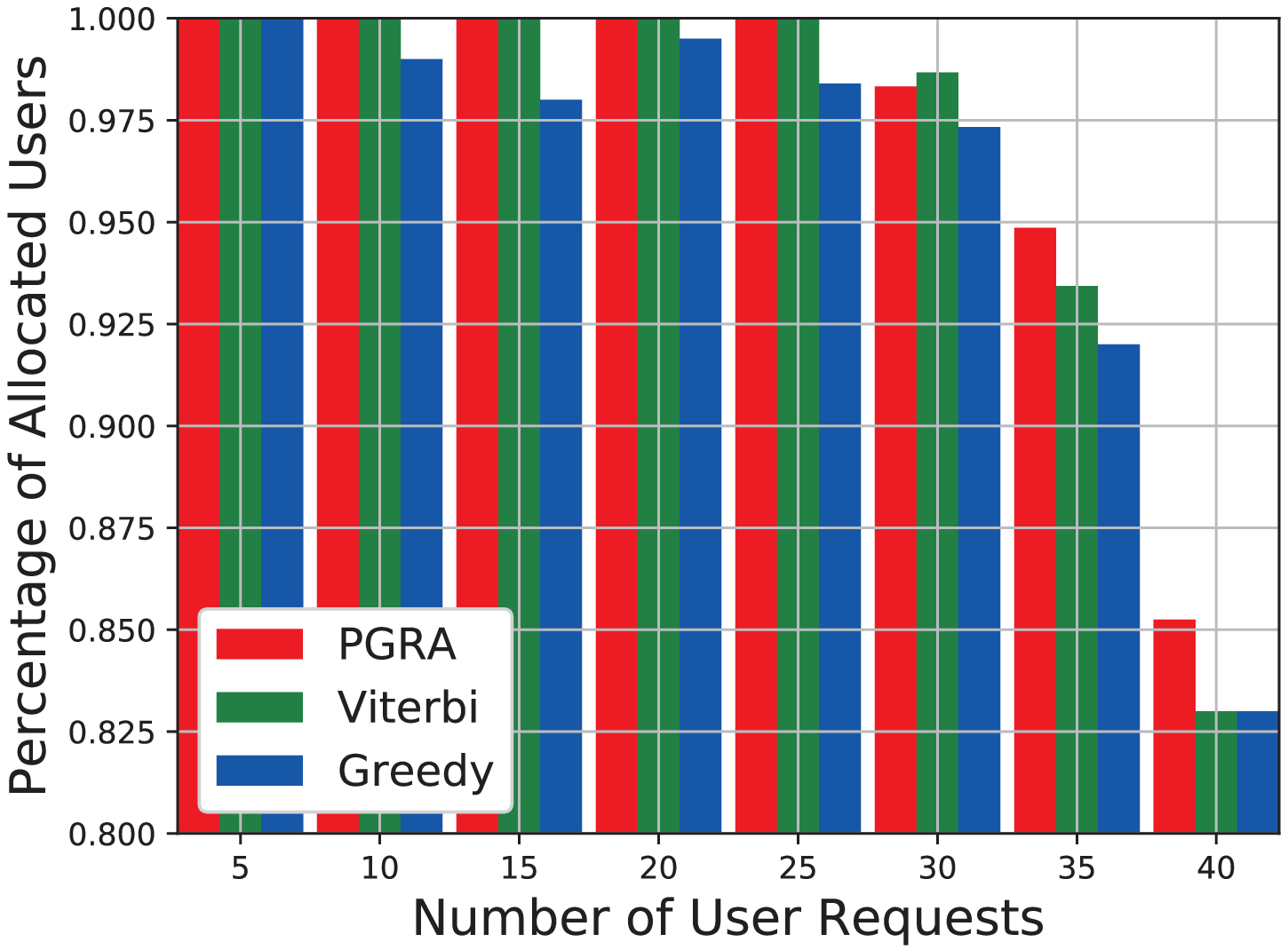}
  \label{Percentages of allocated user requests in a satellite network with N=12}}
  \subfigure[Satellite network with $N=15$]{\includegraphics[width=0.3\textwidth]{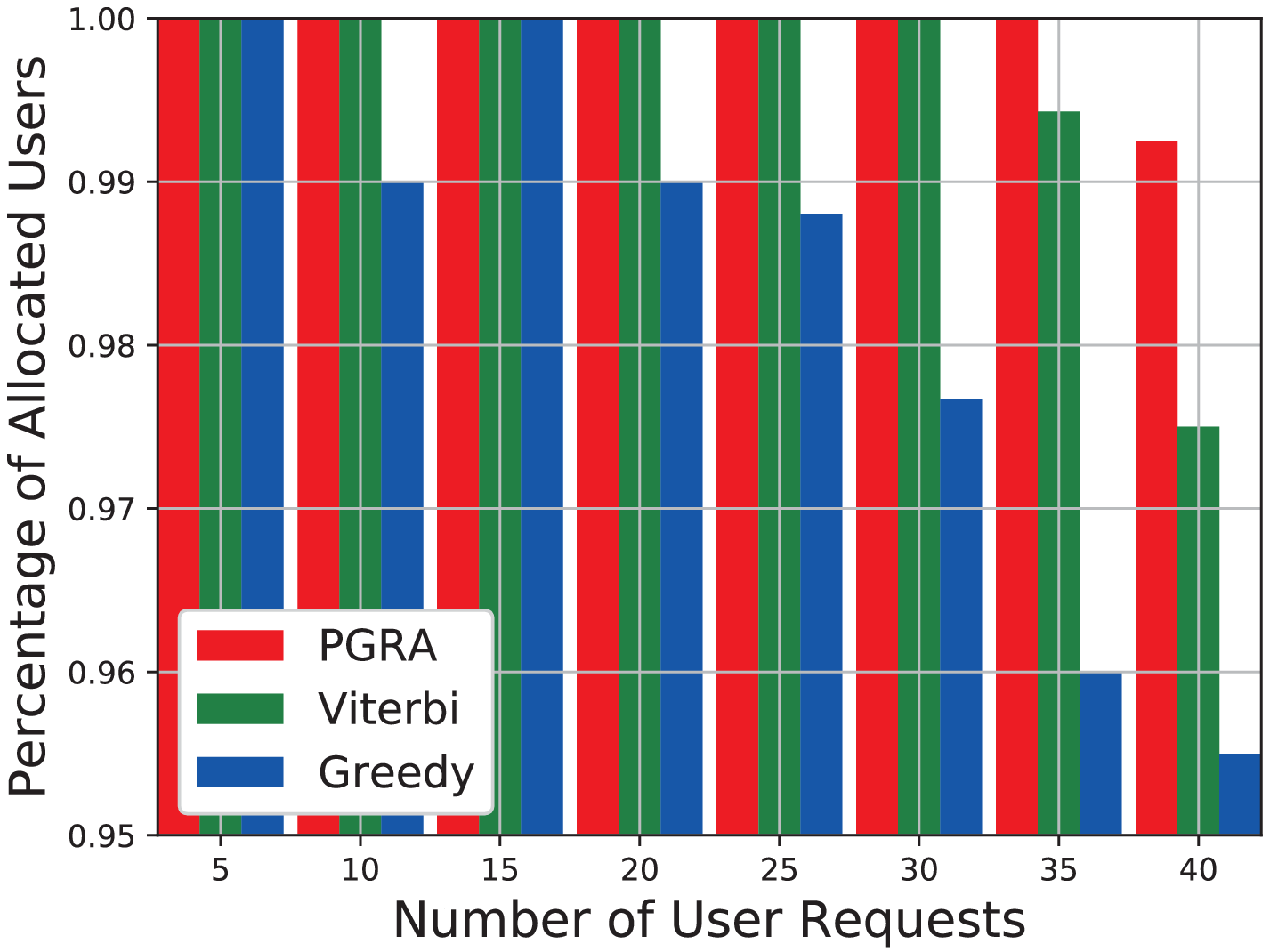}
  \label{Percentages of allocated user requests in a satellite network with N=15}}
  \caption{Percentages of allocated user requests for PGRA, Viterbi, and Greedy in satellite networks with $N=9$, 12, and 15.}
  \label{Percentages of allocated user requests for PGRA, Viterbi, and Greedy in satellite networks}
\end{figure*}

We also provide the percentages of allocated user requests for PGRA, Viterbi, and Greedy in three satellite networks with $N=9$, 12, and 15, as shown in Fig.~\ref{Percentages of allocated user requests for PGRA, Viterbi, and Greedy in satellite networks}. Fig.~\ref{Percentages of allocated user requests in a satellite network with N=9} describes the percentages of allocated user requests for different number of user requests in the satellite network with $N=9$. We can find from Fig.~\ref{Percentages of allocated user requests in a satellite network with N=9} that the percentages of allocated user requests are close to 1 for small number of user requests, e.g., $M=5$, 10, and 15. For small $M$, the resource requirements of user requests are less than the network resource capacities and all of them can be deployed to satellite nodes. As the increasing number in user requests, the available network resources are less and less. When the resource requirements of user requests exceed the network resource capacities, the satellite network can not provide any available resources for new user requests and that can result in a decrease in the percentages of allocated user requests, as shown in Fig.~\ref{Percentages of allocated user requests in a satellite network with N=9}. However, for almost all experiments, the proposed PGRA algorithm performs better than Viterbi and Greedy. On average, the percentage of allocated user requests obtained by the proposed PGRA algorithm improves by $1.52\%$ for Viterbi and $2.56\%$ for Greedy. Fig.~\ref{Percentages of allocated user requests in a satellite network with N=12} and Fig.~\ref{Percentages of allocated user requests in a satellite network with N=15} illustrate the percentages of allocated user requests for different user requests in satellite networks with $N=12$ and 15, respectively, where similar results in Fig.~\ref{Percentages of allocated user requests in a satellite network with N=9} can be obtained. When the number of user requests is small, the satellite network can provide enough available resources for deploying user requests. All of user requests can be deployed to satellite nodes and the percentages of allocated user requests are close to $1$. As $M$ increases, the available resources of satellite nodes are gradually decreasing and the resource requirements of more user requests will be not satisfied. Therefore, the percentage of allocated user requests will reduce as $M$ increases. We can also find from Fig.~\ref{Percentages of allocated user requests in a satellite network with N=12} and Fig.~\ref{Percentages of allocated user requests in a satellite network with N=15} that the performance of the proposed PGRA algorithm is better than that of Viterbi and Greedy. For $N=12$, the performance improvement of the proposed PGRA algorithm is $0.43\%$ for Viterbi and $1.46\%$ for Greedy. For $N=15$, the proposed PGRA algorithm performs better $0.28\%$ than Viterbi and $1.69\%$ than Greedy. Besides, as $N$ increases, more resources of the satellite network are available for deploying user requests and thus the percentages of allocated use requests will increase. For $M=35$, the percentages of allocated user requests obtained by the proposed PGRA algorithm are 0.7314, 0.9485, and 1.0 for satellite networks with $N=9$, 12, and 15, respectively.\par

\subsection{Performance Analysis in On-line Strategy}
\begin{figure*}[tbp]
  \centering
  \subfigure[Average energy cost]{\includegraphics[width=0.3\textwidth]{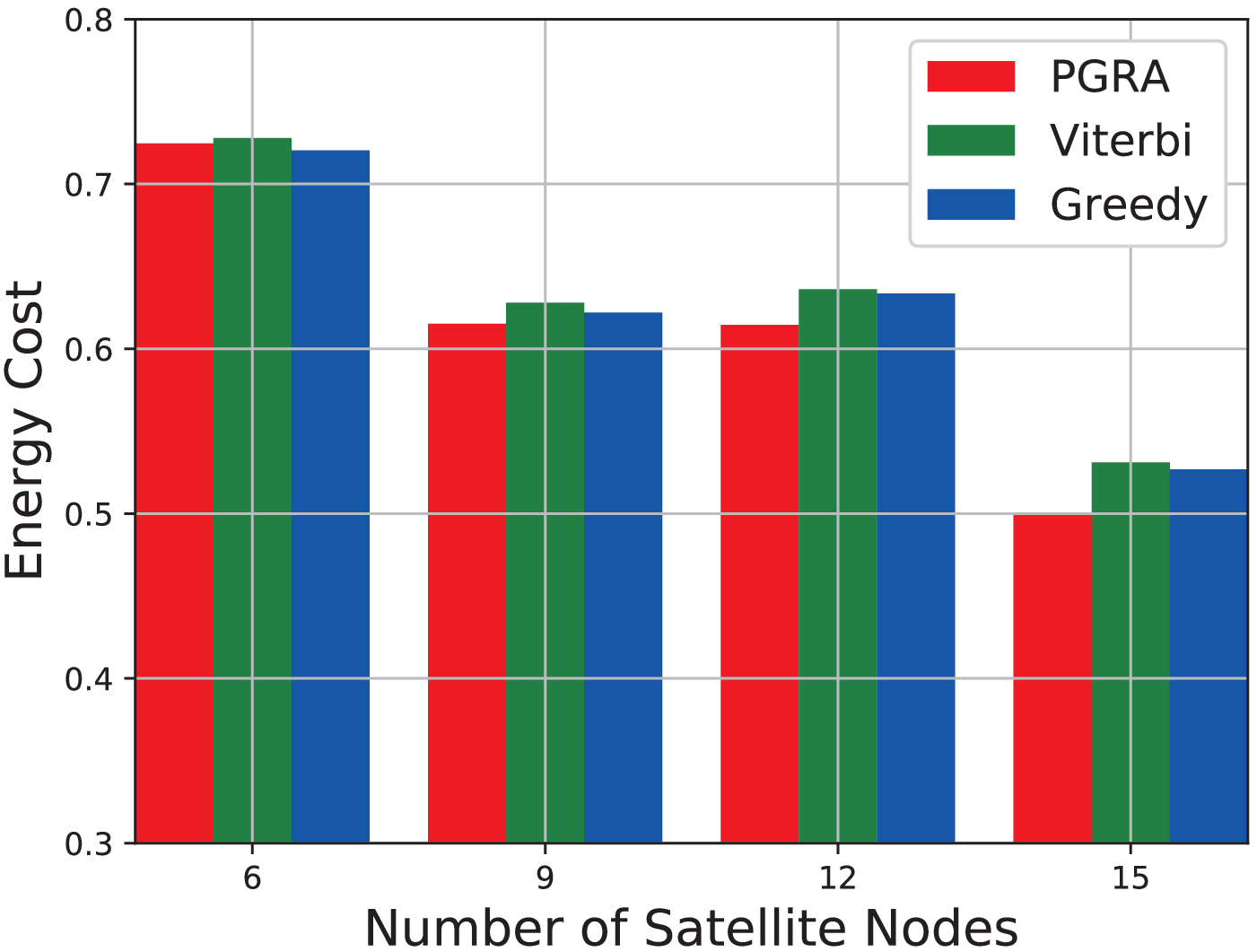}
  \label{Average energy cost in on-line strategy}}
  \subfigure[Average bandwidth cost]{\includegraphics[width=0.3\textwidth]{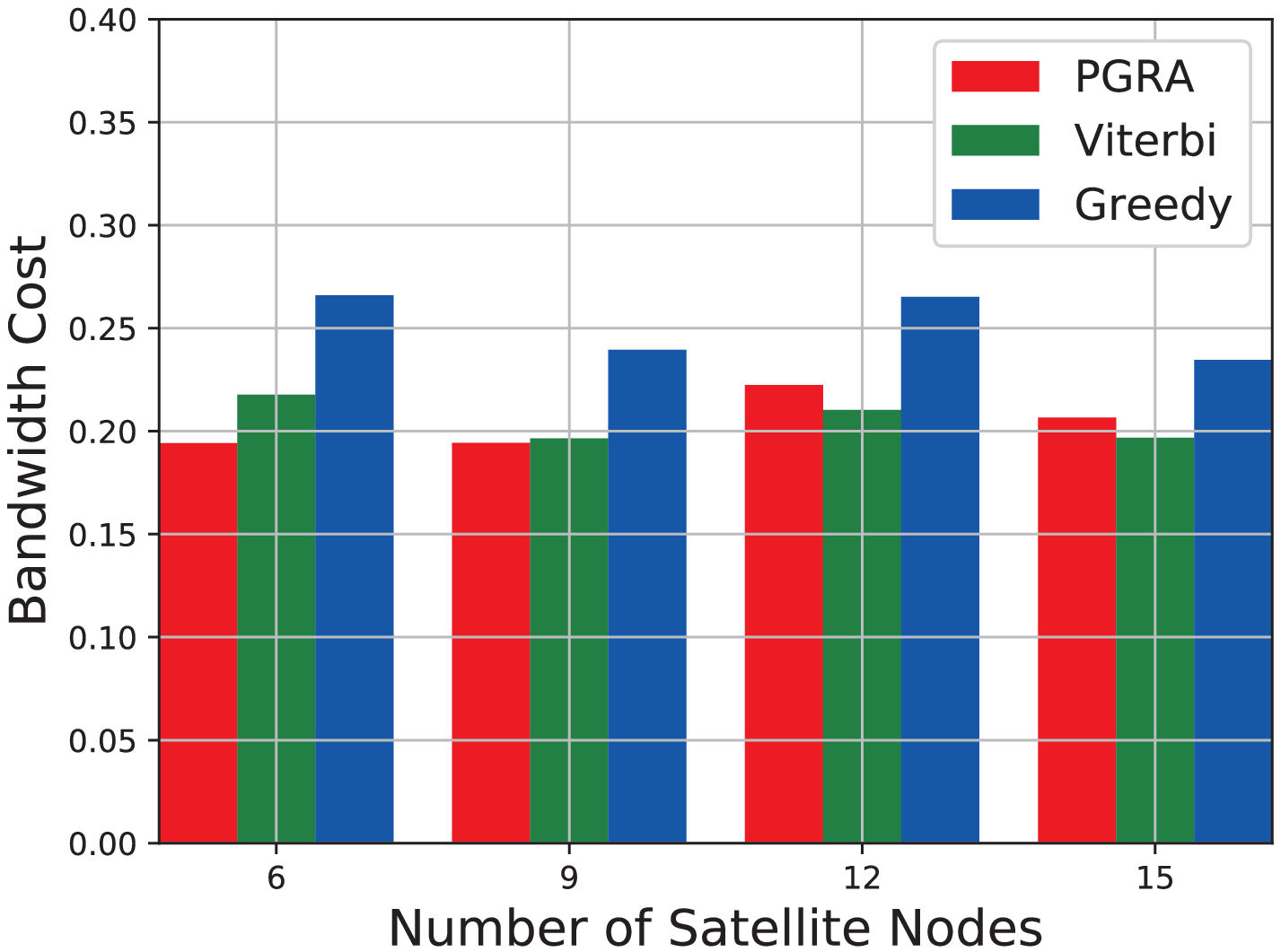}
  \label{Average bandwidth cost in on-line strategy}}
  \subfigure[Average service delay cost]{\includegraphics[width=0.3\textwidth]{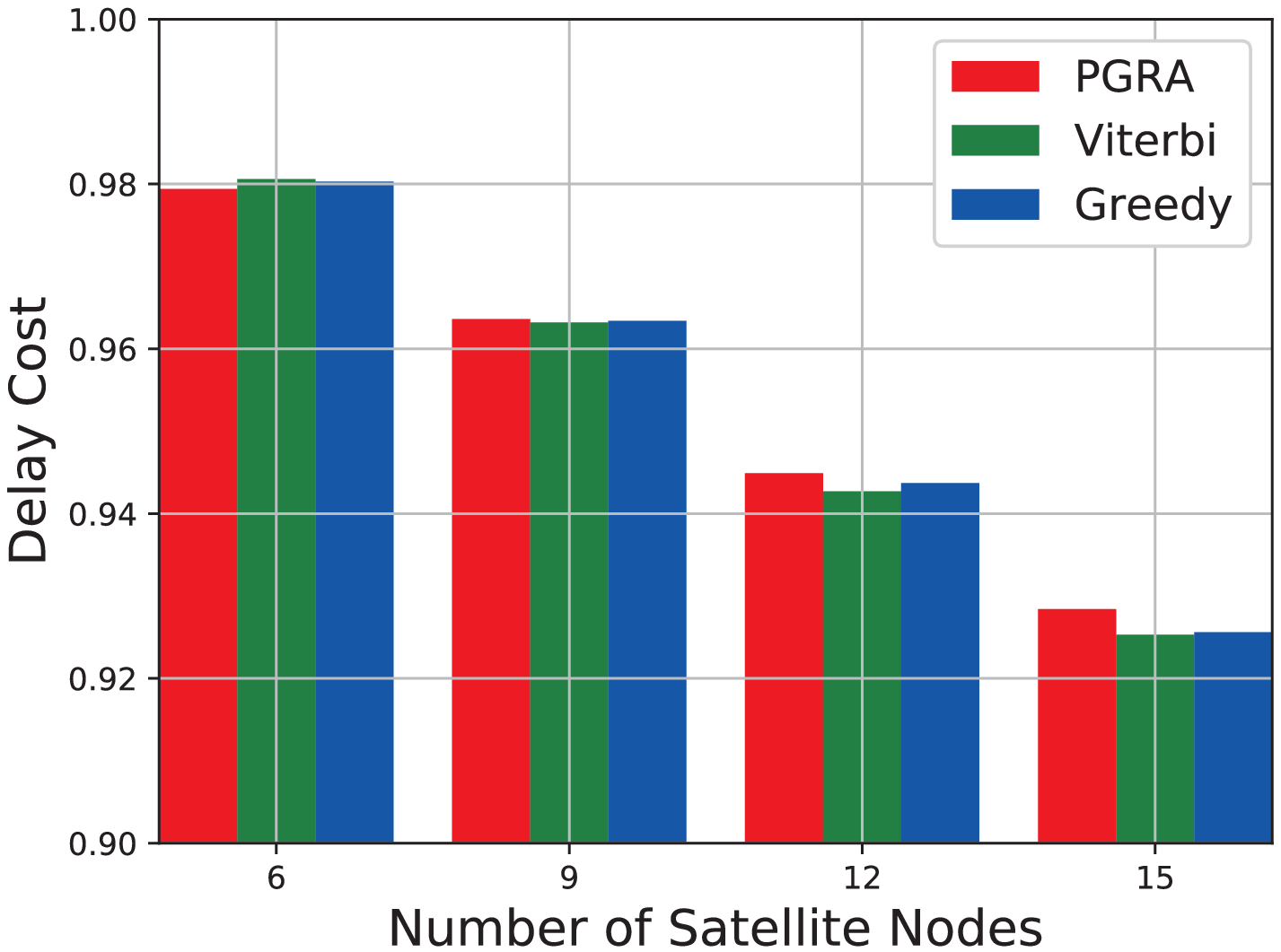}
  \label{Average service delay cost in on-line strategy}}
  \subfigure[Average network payoff]{\includegraphics[width=0.3\textwidth]{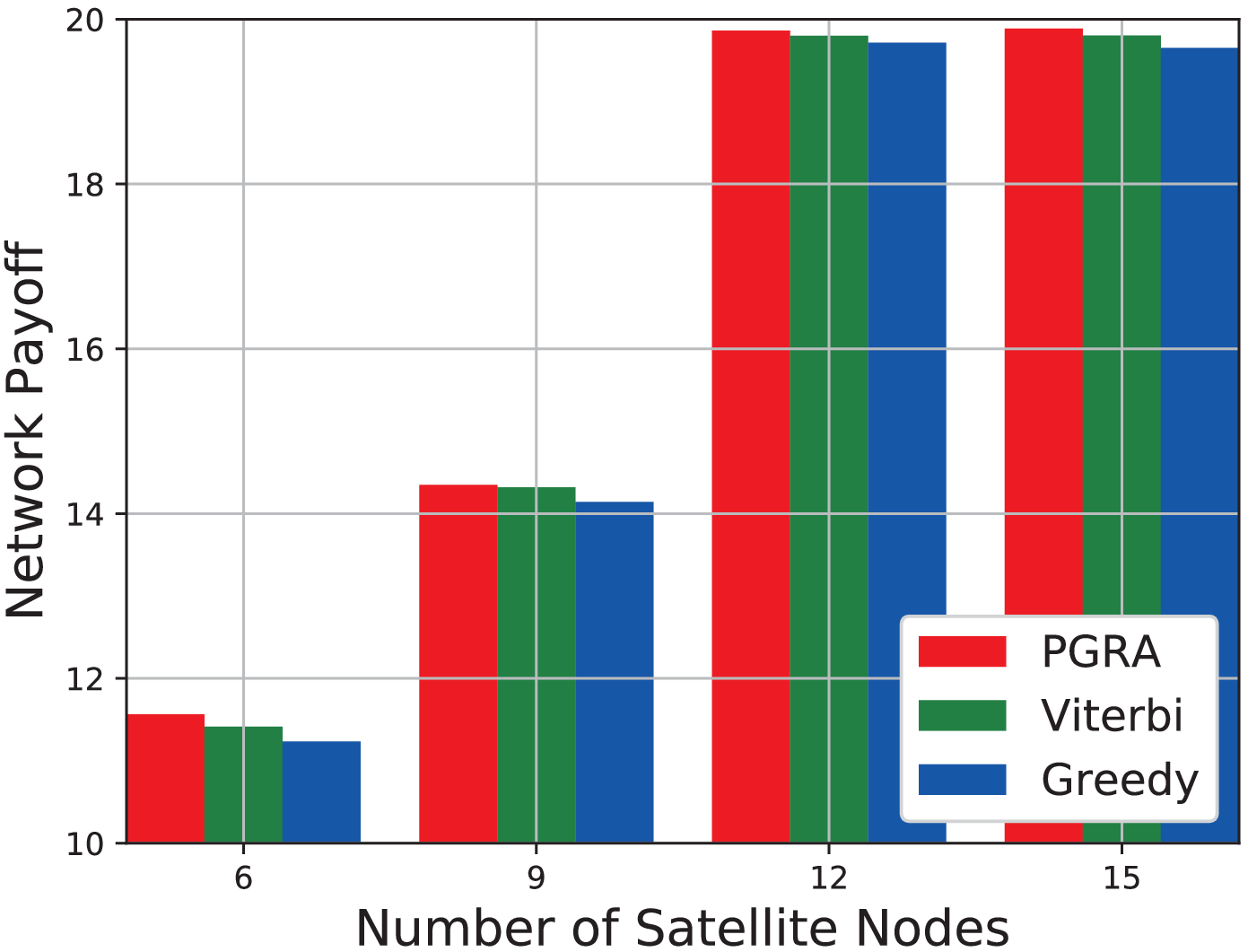}
  \label{Average network payoff in on-line strategy}}
  \subfigure[Average percentage of allocated users]{\includegraphics[width=0.3\textwidth]{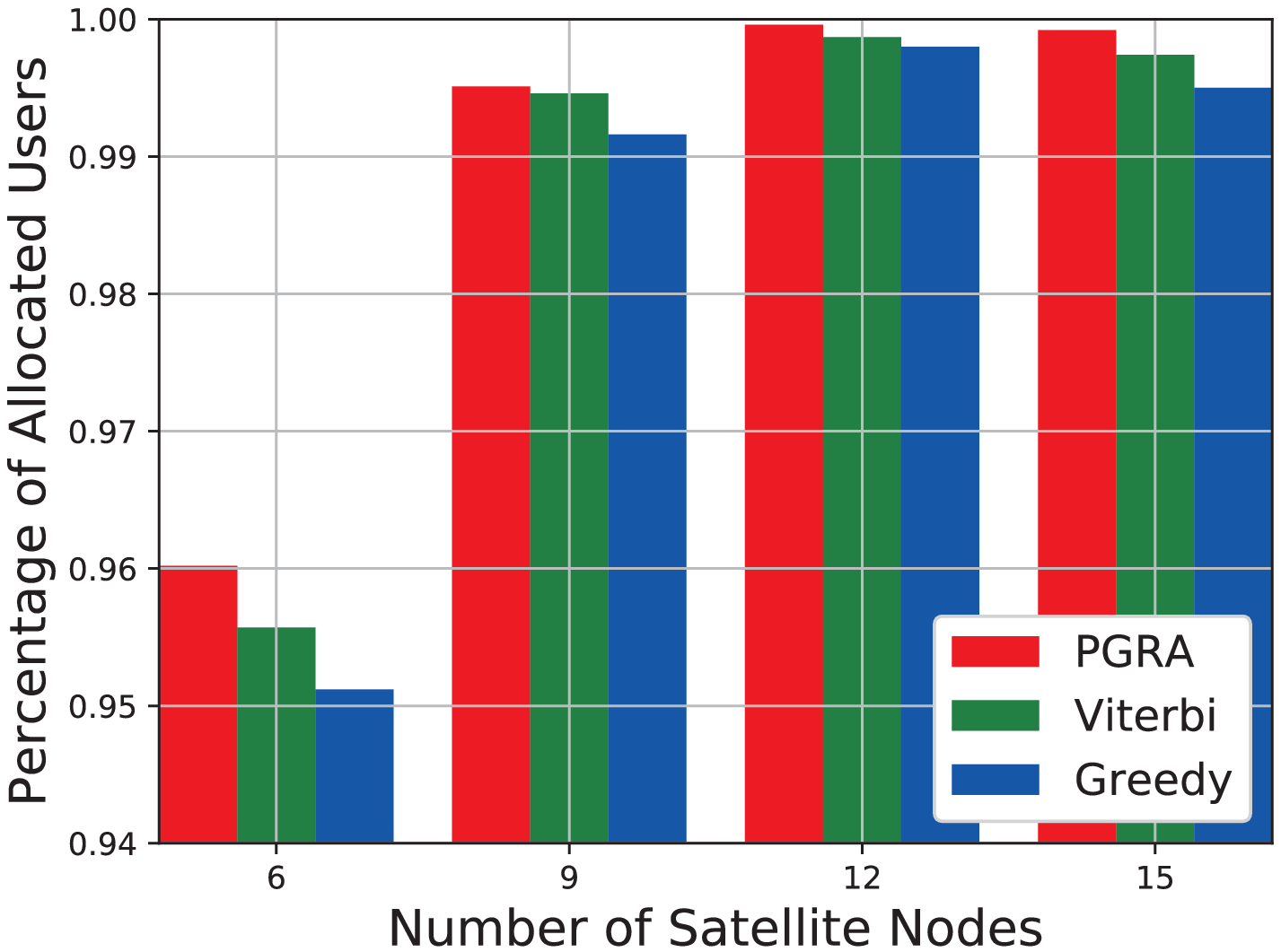}
  \label{Average percentage of allocated users in on-line strategy}}
  \caption{Performance comparison with PGRA, Viterbi, and Greedy in satellite networks with 6, 9, 12, and 15 satellite nodes, respectively.}
  \label{Performance comparison with PGRA, Viterbi, and Greedy in satellite networks with different satellite nodes}
\end{figure*}
To evaluate the effectiveness of the proposed PGRA algorithm for addressing the VNF placement problem in dynamic environment, we make the following experiments in satellite networks with 6, 9, 12, and 15 satellite nodes, respectively. The total number of time slots for each instance is 50 and we randomly generate the number of user requests from 5 to 10 in each time slot. The running periods for user requests can be randomly selected from 1 to 4 time slots. When the running time for a user request is over, the resources used by the user request can free and be available for deploying new user requests in the next time slot. All experiments are performed for 10 times and we obtain the average results.

In the four satellite networks, the average experiment results for deploying user requests in on-line strategy are shown in Fig.~\ref{Performance comparison with PGRA, Viterbi, and Greedy in satellite networks with different satellite nodes}. Fig.~\ref{Average energy cost in on-line strategy} describes the average energy costs for different satellite networks in on-line strategy. We can find from Fig.~\ref{Average energy cost in on-line strategy} that overall the energy costs obtained by the proposed PGRA algorithm for satellite networks with $N=6$, 9, 12, and 15 are less than that of Viterbi and Greedy. For an example of $N=12$, the energy costs for PGRA, Viterbi, and Greedy are 0.6145, 0.6360, and 0.6335, respectively, the energy cost of the proposed PGRA algorithm decreases by $3.38\%$ for Viterbi and $2.99\%$ for Greedy. The average bandwidth costs for different satellite networks in on-line strategy are illustrated in Fig.~\ref{Average bandwidth cost in on-line strategy}. We can find that the bandwidth costs obtained by the proposed PGRA algorithm for $N=6$ and 9 are less than that of Viterbi and Greedy, however, for $N=12$ and 15, are more than that of Viterbi and less than that of Greedy. That is due to the fact that the number of satellite nodes used by user requests can have an impact on the strategy of deploying VNFs. When the number of used satellite nodes is small, adjacent VNFs for a user request may not be deployed to the same satellite node due to the available resource limitation. If adjacent VNFs are deployed two different satellite nodes there will lead to the bandwidth costs as a result of routing traffic flows. In Fig.~\ref{Average service delay cost in on-line strategy}, we describe the service delay costs for different satellite networks. Similar to the results in Fig.~\ref{Average bandwidth cost in on-line strategy}. For $N=12$ and 15, the proposed PGRA algorithm performs better than Viterbi and Greedy in terms of energy consumption, however, the delay costs obtained by the proposed PGRA algorithm are more than that of Viterbi and Greedy, as the result of the tradeoff between bandwidth usage and energy consumption. In addition, we can find from Fig.~\ref{Average service delay cost in on-line strategy} that the average delay cost can decrease as the number of satellite nodes increases. The delay cost obtained by the proposed PGRA algorithm is 0.9793 for $N=6$ and 0.9449 for $N=12$. As the increase in the number of satellite nodes, there are more available resources for allocating user requests. Therefore, the paths with lower delay costs can be chosen to route traffic flows for user requests.

Note that the proposed PGRA outperforms both Viterbi and Greedy for network payoffs and percentages of allocated users, where the average network payoffs and percentages of allocated user requests are shown in Fig.~\ref{Average network payoff in on-line strategy} and Fig.~\ref{Average percentage of allocated users in on-line strategy}, respectively. From Fig.~\ref{Average network payoff in on-line strategy}, we can find that the network payoffs obtained by the proposed PGRA algorithm for satellite networks with $N=6$, 9, 12, and 15 are better than that of Viterbi and Greedy. For these four satellite networks, the network payoff improvement of the proposed PGRA algorithm is $1.31\%$, $0.21\%$, $0.31\%$, and $0.42\%$ for Viterbi, and $2.92\%$, $1.45\%$, $0.73\%$, and $1.18\%$ for Greedy, respectively. Furthermore, we can also find that the network payoff can increase as the number of satellite nodes increases. That is the fact that as available resources in satellite networks increase more user requests can be deployed to satellite nodes. For example, the network payoffs obtained by the proposed PGRA algorithm, in satellite networks with $N=6$, 9, 12, and 15, are 11.5633, 14.3481, 19.8609, and 19.8865, respectively. The average percentages of allocated user requests are shown in Fig.~\ref{Average percentage of allocated users in on-line strategy}. The proposed PGRA algorithm performs better than Viterbi and Greedy. For satellite networks with $N=6$, 9, 12, and 15, the percentages of allocated user requests obtained by the proposed PGRA algorithm increase by $0.47\%$, $0.04\%$, $0.09\%$, and $0.18\%$ for Viterbi, and $0.93\%$, $0.34\%$, $0.16\%$, and $0.42\%$ for Greedy, respectively. As the number of satellite nodes increases, we can also find that the percentage of allocated user requests increases. For the proposed PGRA algorithm, the percentages of allocated user requests in satellite networks with $N=6$, 9, 12, and 15 are $96.01\%$, $99.50\%$, $99.96\%$, and $99.92\%$, respectively. From Fig.~\ref{Performance comparison with PGRA, Viterbi, and Greedy in satellite networks with different satellite nodes}, we can demonstrate the effectiveness of the proposed PGRA algorithm compared with Viterbi and Greedy in dynamic environment.

\section{Conclusion}\label{Conclusion}
In this paper, we study the VNF placement problem by a potential game in satellite edge computing. Our aim is to maximize the number of allocated user requests with minimum overall deployment cost, which include energy, bandwidth, and service delay costs. We prove the VNF placement problem to be NP-hard and formulate the problem as a potential game with maximum network payoff. Each user request can make the deployment decision in a self-interested way and all user requests can optimize their strategies by competing with each other in a distributed manner. Considering that a potential game admits at least a Nash equilibrium we implement a decentralized resource allocation algorithm based on a potential game to find an approximate solution.

To evaluate the effectiveness of the proposed PGRA algorithm, we first discuss the influence of system parameters on the proposed PGRA algorithm performance by the Taguchi method. Then we make the experiments for different number of user requests in satellite networks with 6, 9, 12, and 15 satellite nodes and compare the simulation results with two baseline algorithms of Viterbi and Greedy. For example, in the case of $N=12$, the average network payoff obtained by the proposed PGRA algorithm increases by $0.96\%$ for Viterbi and $2.08\%$ for Greedy, the average percentage of allocated users obtained by the proposed PGAR algorithm improves by $0.43\%$ for Viterbi and $1.46\%$ for Greedy. In addition, we investigate the performance of the proposed PGRA algorithm in dynamic environment. In dynamic environment, for $N=12$, the proposed PGRA algorithm improves by $0.31\%$ for Viterbi and $0.73\%$ for Greedy in terms of network payoffs, the percentage of allocated users obtained by the proposed PGAR algorithm increases by $0.09\%$ for Viterbi and $0.16\%$ for Greedy. All the simulation results show that the proposed PGRA algorithm is an effective approach for addressing the VNF placement in satellite edge computing and performs better than Viterbi and Greedy.


%

%

%
%

\ifCLASSOPTIONcaptionsoff
  \newpage
\fi



\bibliographystyle{IEEEtran}
\bibliography{Virtual_Network_Function_Placement_in_Satellite_Edge_Computing_with_a_Potential_Game_Approach}
%
%
%

%



\begin{IEEEbiography}[{\includegraphics[width=1in,height=1.25in,clip,keepaspectratio]{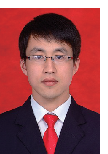}}]{Xiangqiang Gao} received the B.Sc. degree in school of electronic engineering from Xidian University and the M.Sc. degree from Xi\textquoteright an Microelectrinics Technology Institute, Xi\textquoteright an, China, in 2012 and 2015, respectively. He is currently pursuing the Ph.D. degree with the School of Electronic and Information Engineering, Beihang University, Beijing, China. His research interests include rateless codes, software defined network and network function virtualization.\par
\end{IEEEbiography}

\begin{IEEEbiography}[{\includegraphics[width=1in,height=1.25in,clip,keepaspectratio]{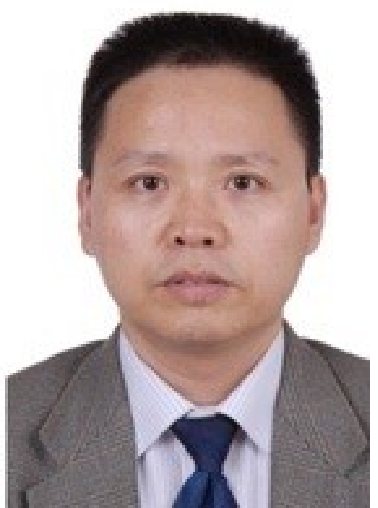}}]{Rongke Liu} received the B.Sc. degree in electronic engineering and Ph.D. degree in information and communication engineering from Beihang University, Beijing, China, in 1996 and 2002, respectively. From 2006 to 2007, he was a visiting professor at Florida Institute of Technology, Florida. In August, 2015, he visited the university of Tokyo as a senior visiting scholar. He is a Full Professor with the School of Electronic and Information Engineering in Beihang University, specializing in the fields of information and communication engineering. He has authored or co-authored more than 100 papers in journals and conferences, and edited four books. His current research interests include multimedia computing and space information network. He is a Member of the IEEE and ACM. Dr. Liu was one of the winners of education ministry\textquoteright s New Century Excellent Talents supporting plan in 2012.\par
\end{IEEEbiography}
\begin{IEEEbiography}[{\includegraphics[width=1in,height=1.25in,clip,keepaspectratio]{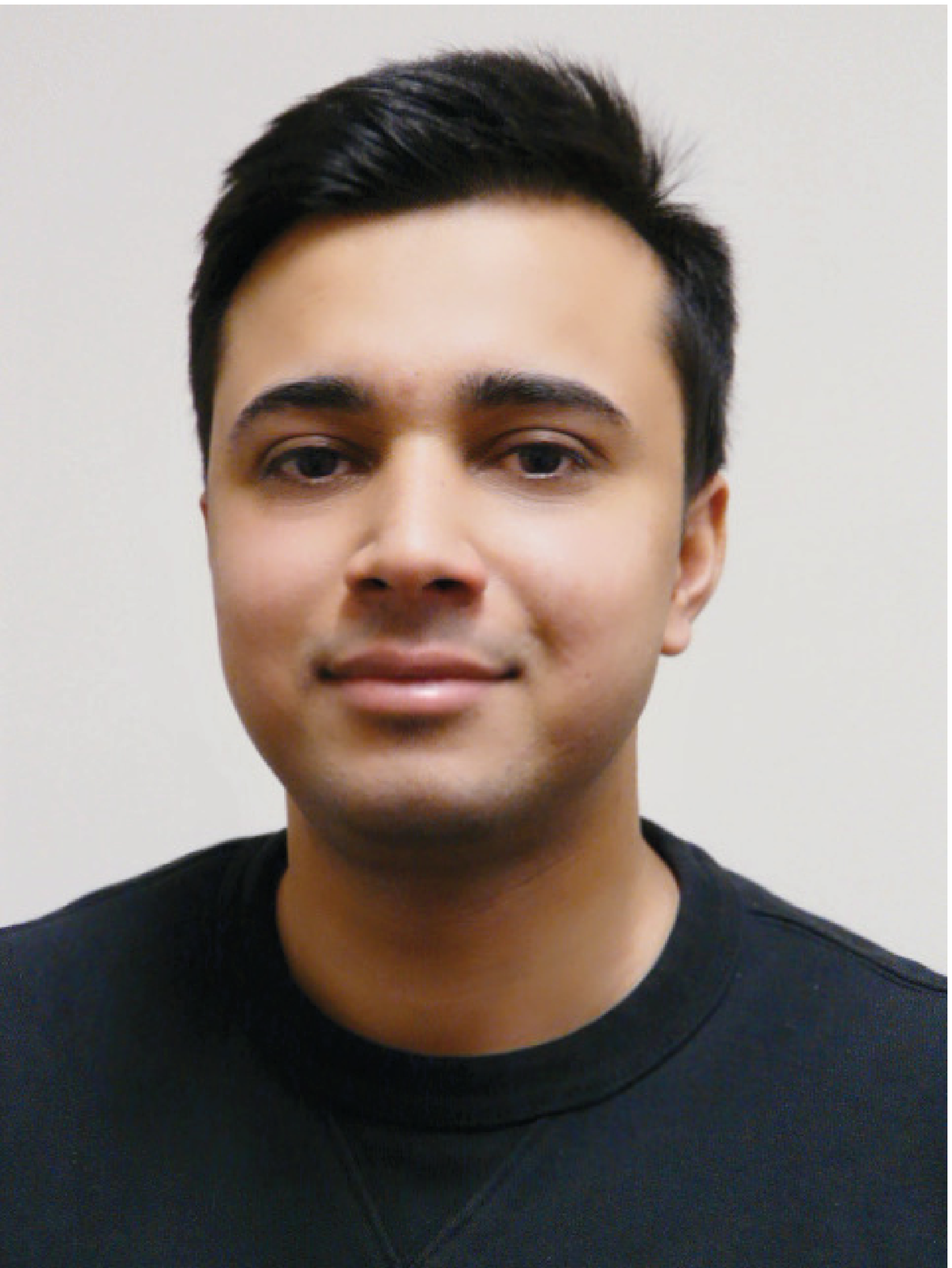}}]{Aryan Kaushik} is currently a Research Fellow at the Department of Electronic and Electrical Engineering, University College London (UCL), United Kingdom, from Feb. 2020. He received PhD in Communications Engineering at the Institute for Digital Communications, School of Engineering, The University of Edinburgh, United Kingdom, in Jan. 2020. He received MSc in Telecommunications from The Hong Kong University of Science and Technology, Hong Kong, in 2015. He has held visiting research appointments at the Wireless Communications and Signal Processing Lab, Imperial College London, UK, from 2019-20, the Interdisciplinary Centre for Security, Reliability and Trust, University of Luxembourg, Luxembourg, in 2018, and the School of Electronic and Information Engineering, Beihang University, China, from 2017-19. His research interests are broadly in signal processing, radar, wireless communications, millimeter wave and multi-antenna communications.\par
\end{IEEEbiography}

\vfill


\end{document}